\documentclass{article}

\usepackage{arxiv}
\usepackage{amsthm}
\usepackage[utf8]{inputenc} 
\usepackage[T1]{fontenc}    
\usepackage{hyperref}       
\usepackage{url}            
\usepackage{booktabs}       
\usepackage{amsfonts}       
\usepackage{nicefrac}       
\usepackage{microtype}      
\usepackage{lipsum}
\usepackage{graphicx}
\graphicspath{ {./images/} }
\usepackage{appendix}

\usepackage{amssymb}
\usepackage{pifont}
\RequirePackage{amsmath,amsfonts,amssymb}
\RequirePackage{amsmath}

\RequirePackage{amsmath}
\newtheorem{definition}{Definition}
\newtheorem{proposition}{Proposition}
\newtheorem{assumption}{Assumption}
\newtheorem{theorem}{Theorem}
\newtheorem{corollary}{Corollary}

\RequirePackage{times}
\RequirePackage{bm}

\RequirePackage{amsmath}

\RequirePackage{hyperref} 
\hypersetup{colorlinks,
citecolor=blue,
linkcolor=blue,
urlcolor=blue
}

\usepackage{thmtools}
\RequirePackage{float}
\RequirePackage{perpage}
\MakeSorted{figure}
\MakeSorted{table}

\RequirePackage{comment}

\newtheorem{example}{Example}

\usepackage{booktabs}

\newcommand{\mathbbm}[1]{\text{\usefont{U}{bbm}{m}{n}#1}} 

\usepackage{stackengine}
\def\defeq{\mathrel{\ensurestackMath{\stackon[1pt]{=}{\scriptscriptstyle\Delta}}}}
\RequirePackage{url}
\RequirePackage{chapterbib}

\RequirePackage{color}
\RequirePackage{tikz}
\tikzset{%
mynode/.style={circle,minimum width=.5ex, fill=none,draw}, 
myfillnode/.style={circle,minimum width=.5ex, fill=lightgray,draw}, 
}
\RequirePackage{amssymb}

\newcommand{\indep}{\perp \!\!\! \perp}
\RequirePackage{amsmath}               
\RequirePackage{lscape}
\RequirePackage{color}
\RequirePackage{tikz}
\RequirePackage{times}
\usepackage{pifont}
\newcommand{\xmark}{\ding{55}}
\RequirePackage{bm}
\usepackage{natbib}
\newenvironment{customeg}[1]
  {\renewcommand{\theexample}{#1}
   \begin{example}}
  {\end{example}}

\usepackage{soul}
\newcommand{\jin}[1]{\textcolor{blue}{[\textbf{Jin}: #1]}}
\newcommand{\jina}[1]{\textcolor{blue}{#1}}
\newcommand{\yuta}[1]{\textcolor{red}{#1}}
\newcommand{\removed}[1]{}
\def\bSig\mathbf{\Sigma}

\title{Cumulative Natural Direct and Indirect Effects\\ for Causal Mediation Analysis}

\author{
Yuta Kawakami \\
Mohamed bin Zayed University of Artificial Intelligence\\
    \And
 Jin Tian \\
Mohamed bin Zayed University of Artificial Intelligence\\
}

\begin{document}
\maketitle
\begin{abstract}
Causal mediation analysis provides a fundamental framework for quantifying the contributions of different pathways from a treatment $X$ to an outcome $Y$ through a mediator. The natural direct and indirect effects (NDE and NIE) are widely used to decompose the total effect. In this paper, we observe that NDE and NIE can give rise to paradoxical interpretations due to their failure to satisfy two desirable properties of interpretable causal effects: skew-symmetry and additivity. To address these limitations, we introduce new measures of direct and indirect effects for continuous treatments, termed the cumulative natural direct and indirect effects (CNDE and CNIE), constructed by decomposing local causal effects $\mathbb{E}[\partial_xY_{x}]$ into local direct and indirect effects. CNDE and CNIE yield a decomposition of the total effect that preserves both skew-symmetry and additivity. We further extend this framework to ordinal treatments by defining discrete analogues of the cumulative effects over ordered treatment levels that preserve these structural properties. We establish decomposition and identification results for the proposed measures under standard causal assumptions. We illustrate their behavior, in comparison with NDE and NIE, using canonical linear mediation models with interaction and a real-world dataset.
\end{abstract}


\section{Introduction}

Causal mediation analysis provides a fundamental framework for  revealing the strength of different pathways linking a treatment \(X\) to an outcome \(Y\) through a mediator \(M\) \citep{Wright1921,Wright1934,Baron1986,Imai2010a,Imai2010b,TchetgenTchetgen2012}. It constitutes a central topic in modern causal inference \citep{Pearl09,Vanderweele2015,Hernan2023}. Within this framework, the total effect  of a treatment is decomposed into direct and indirect components, capturing pathways that bypass or operate through the mediator. In particular, the natural direct and indirect effects (NDE and NIE), grounded in the potential outcomes and structural causal model frameworks \citep{Robins1992,Pearl2001},  provide a principled decomposition of the total effect. 

The total effect (TE),
$
\text{TE}(x'',x';Y) = \mathbb{E}[Y_{x''}] - \mathbb{E}[Y_{x'}],
$ 
satisfies two fundamental properties:
\begin{equation}
\begin{aligned}
&\text{\bf Skew-symmetry:} \quad 
\text{TE}(x'',x';Y) = -\text{TE}(x',x'';Y),
\end{aligned}
\end{equation}
\begin{equation}
\begin{aligned}
&\text{\bf Additivity:} \quad 
\text{TE}(x''',x';Y)
= \text{TE}(x''',x'';Y) + \text{TE}(x'',x';Y),
\end{aligned}
\end{equation}
for any \(x', x'', x'''\). These properties are intuitive: reversing the treatment contrast reverses the sign of the effect, and effects over adjacent contrasts add up consistently along a path. Such natural algebraic properties are fundamental for interpretable causal effects. 

In contrast, the NDE and NIE do not, in general, satisfy skew-symmetry or additivity. As a result, their decomposition of the total effect can exhibit unintuitive behavior. For example, reversing the treatment contrast may alter the relative contributions of direct and indirect effects, and decompositions over sub-intervals may not align with the decomposition over the full interval. Such phenomena can lead to what we term \emph{paradoxical interpretations} in mediation analysis (see Section~\ref{sec-paradox} for detailed examples). While these paradoxes do not invalidate NDE and NIE as causal quantities, they complicate interpretation and may hinder their use in practice.

To address these issues,  in this paper, we introduce new measures of direct and indirect effects that preserve both skew-symmetry and additivity. We begin with the case of a continuous treatment \(X\). Rather than decomposing global contrasts directly, we introduce a local decomposition of the causal effect, based on the derivative \(\mathbb{E}[\partial_x Y_x]\), into local direct and indirect components. By integrating these local effects along the treatment path, we define the \emph{cumulative natural direct and indirect effects} (CNDE and CNIE). These measures yield a pathwise decomposition of the total effect that satisfies both skew-symmetry and additivity. 
Although prior work has studied the estimation of mediation effects under continuous treatments \citep{Wang2016,Huber2020,Xu2021,Zhang2025}, existing approaches do not address the structural limitations of NDE and NIE. To the best of our knowledge, there has been no systematic development of direct and indirect effect measures that are specifically designed to preserve fundamental properties of skew-symmetry and additivity.

We further extend our framework to ordinal treatments. In this setting, we construct discrete analogues of the cumulative effects by aggregating local increments along ordered treatment levels. We  consider a  skew-symmetric version of NDE and NIE and  propose direct and indirect effect formulations that additionally satisfy additivity.

For each type of proposed measures, we establish decomposition results showing that the total effect can be expressed as the sum of the proposed direct and indirect effects, and we provide identification results under 
standard causal assumptions.
We illustrate the behavior of these measures, in comparison with NDE and NIE, using canonical linear mediation models with interaction. Finally, we demonstrate their practical relevance through an application to a real-world dataset.

\begin{table}[t]
\centering
\caption{Properties of direct and indirect effect measures considered in this paper.
\checkmark~indicates the property holds, and \xmark~indicates it does not.
Additivity$^*$ denotes additivity holds only for monotone triplets 
($x' < x'' < x'''$ or $x' > x'' > x'''$).}
\label{tab:properties}
\begin{tabular}{lccc}
\toprule
Measure & Skew-symmetry & Additivity & Treatment type \\
\midrule
NDE     & \xmark & \xmark          & Any        \\
NIE     & \xmark & \xmark         & Any        \\
S-NDE   & \checkmark & \xmark      & Any        \\
S-NIE   & \checkmark & \xmark      & Any        \\
CNDE    & \checkmark & \checkmark  & Continuous \\
CNIE    & \checkmark & \checkmark  & Continuous \\
CNDE-O  & \xmark & $\checkmark^*$  & Ordinal    \\
CNIE-O  & \xmark & $\checkmark^*$  & Ordinal    \\
S-CNDE-O & \checkmark & $\checkmark^*$ & Ordinal \\
S-CNIE-O & \checkmark & $\checkmark^*$ & Ordinal \\
\bottomrule
\end{tabular}
\end{table}

This paper makes the following contributions: 
\begin{itemize}
\item We identify fundamental structural limitations of NDE and NIE, namely their failure to satisfy skew-symmetry and additivity, and show how these limitations can lead to paradoxical interpretations in mediation analysis.

\item For continuous treatments, we introduce the cumulative natural direct and indirect effects (CNDE and CNIE), which satisfy both skew-symmetry and additivity, yielding a coherent decomposition of the total effect.

\item For ordinal treatments, we develop cumulative direct and indirect effect measures, S-CNDE-O and S-CNIE-O, based on skew-symmetric versions of NDE and NIE, termed S-NDE and S-NIE. The proposed measures preserve both skew-symmetry and additivity.

\end{itemize}
Table \ref{tab:properties} summarizes the properties of the proposed  direct and indirect effect measures.

\section{Notation and Background}

In this section, we introduce the notation and review previous studies on causal mediation analysis.
We represent each variable with a capital letter $(X)$ and its realized value with a small letter $(x)$.
Denote $\Omega_Y$ as the domain of $Y$,
$\mathbb{E}[Y]$ as the expectation of $Y$, $\mathbb{P}(Y> y)$ as the cumulative distribution function (CDF) of a continuous variable $Y$, and $\mathfrak{p}(Y=y)$ as the probability density function (PDF) of a continuous variable $Y$.
We denote $Y \indep M|X$ if $Y$ and $M$ are conditionally independent given $X$.
\begin{definition}[Skew-symmetry and additivity]
Let $\phi : \Omega_X \times \Omega_X \to \mathbb{R}$. We say that $\phi$ is \emph{skew-symmetric} if $\phi(x'',x')=-\phi(x',x'')$, and \emph{additive} if $\phi(x''',x')=\phi(x''',x'')+\phi(x'',x')$,  for any $x',x'',x'''\in \Omega_X$.
\end{definition}

{\bf Structural Causal Models.}
We use the language of Structural Causal Models (SCM) as our basic semantic and inferential framework \citep{Pearl09}.
An SCM $\mathcal{M}$ is a tuple $\left<{\boldsymbol V},{\boldsymbol U}, \mathcal{F}, \mathbb{P}_{\boldsymbol U} \right>$, where ${\boldsymbol U}$ is a set of exogenous (unobserved) variables following a joint distribution $\mathbb{P}_{\boldsymbol U}$, and ${\boldsymbol V}$ is a set of endogenous (observable) variables whose values are determined by structural functions $\mathcal{F}=\{f_{V_i}\}_{V_i \in {\boldsymbol V}}$ such that $v_i:= f_{V_i}({\mathbf{pa}}_{V_i},{\boldsymbol u}_{V_i})$ where ${\mathbf{PA}}_{V_i} \subseteq {\boldsymbol V}$ and $U_{V_i} \subseteq {\boldsymbol U}$. 
Each SCM $\mathcal{M}$ induces an observational distribution $\mathbb{P}_{\boldsymbol V}$ over ${\boldsymbol V}$, and a causal graph $G(\mathcal{M})$ over ${\boldsymbol V}$ in which there exists a directed edge from every variable in ${\mathbf{PA}}_{V_i}$ to $V_i$.
An intervention of setting a set of endogenous variables ${\boldsymbol X}$ to constants ${\boldsymbol x}$, denoted by $do({\boldsymbol x})$, replaces the original equations of ${\boldsymbol X}$ by the constants ${\boldsymbol x}$ and induces a \textit{sub-model}  $\mathcal{M}_{{\boldsymbol x}}$.
We denote the potential outcome $Y$ under intervention $do({\boldsymbol x})$ by $Y_{{\boldsymbol x}}({\boldsymbol u})$, which is the solution of $Y$ in the sub-model $\mathcal{M}_{{\boldsymbol x}}$ given ${\boldsymbol U}={\boldsymbol u}$. 

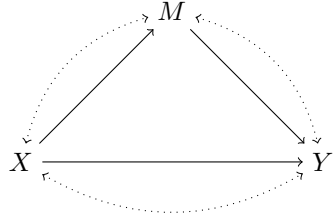
\begin{figure}[tb]
\centering
\scalebox{1}{
\begin{tikzpicture}
\node (z) at (0,2) {$M$};
\node (y) at (2,0) {$Y$};
\node (x) at (-2,0) {$X$};

\path (x) edge[->] (y);
    
\path (z) edge[->] (y);
\path (x) edge[->] (z);

\path (x) edge[<->,dotted,bend left] (z);
\path (z) edge[<->,dotted,bend left] (y);
\path (x) edge[<->,dotted,bend right] (y);
    

\end{tikzpicture}
}
\vspace{-0cm}
\caption{The causal graph representing SCM $\mathcal{M}_{N}$.}
\label{DAGN1}
\end{figure}

{\bf Causal Mediation Analysis.}
Let $X$ be a continuous treatment variable, $Y$ be a continuous outcome, and $M$ be a continuous mediation variable. 
We consider the following nonparametric SCM, $\mathcal{M}_{N}$:
\begin{equation}\label{eq-scmmn}
\begin{gathered}
  \mathcal{M}_{N}:  X:=f_X(\epsilon_X),\ \ 
    M:=f_{M}(X,\epsilon_{M}),\ \
    Y:=f_Y(X,M,\epsilon_Y),
   \end{gathered}
\end{equation}
where $\epsilon_X$, $\epsilon_{M}$, and $\epsilon_Y$ are latent exogenous variables following a joint distribution $\mathbb{P}$. 
Figure \ref{DAGN1} shows a causal graph representing $\mathcal{M}_{N}$, where bidirected edges indicate unmeasured exogenous confounders that affect the variables.

\citet{Pearl2001} defines the total, controlled direct, natural direct, and indirect effects as follows: 
\begin{definition}[TE, CDE, NDE, and NIE] 
\label{def1}
\mbox{}
\begin{enumerate}
\setlength{\itemsep}{2pt}
\setlength{\parskip}{4pt}
    \item Total Effect (TE):
    $\text{TE}(x'',x';Y)\defeq\mathbb{E}[Y_{x''}]-\mathbb{E}[Y_{x'}]$
    \item Controlled Direct Effect (CDE):
    $\text{CDE}(x'',x',m;Y)\defeq\mathbb{E}[Y_{x'',{m}}]-\mathbb{E}[Y_{x',{m}}]$
    \item Natural Direct Effect (NDE):
    $\text{NDE}(x'',x';Y)\defeq\mathbb{E}[Y_{x'',{M_{x'}}}]-\mathbb{E}[Y_{x'}]$
    \item Natural Indirect Effect (NIE):
    $\text{NIE}(x'',x';Y)\defeq\mathbb{E}[Y_{x',{M_{x''}}}]-\mathbb{E}[Y_{x'}]$ 
\end{enumerate}
for any $x', x'' \in \Omega_X$ and $m \in \Omega_M$.
\end{definition}
TE satisfies skew-symmetry and additivity: $\text{TE}(x'',x';Y)=-\text{TE}(x',x'';Y)$ and $\text{TE}(x''',x';Y)=\text{TE}(x''',x'';Y)+\text{TE}(x'',x';Y)$.
CDE also satisfies skew-symmetry and additivity; however, it depends on the values $m$ and has no corresponding indirect effect. 
On the other hand, neither NDE nor NIE satisfies the skew-symmetry or additivity, except in special cases, e.g., under Baron and Kenny's linear no-interaction model (see Example~\ref{eg-m1}). 
TE can be decomposed by 
\begin{equation}
\text{TE}(x'',x';Y)=\text{NDE}(x'',x';Y)-\text{NIE}(x',x'';Y).
\end{equation}


These direct and indirect effects may be identified from observational distributions under various settings \citep{Pearl2001,Avin2005,Shpitser2008,Shpitser2013,Malinsky2019}.
NDE and NIE are both identifiable through identifying the expectations of counterfactuals of the form $\mathbb{E}[Y_{x'',{M_{x'}}}]$.
The counterfactual $Y_{x'',{M_{x'}}}$ denotes the value of $Y$ under $do(X=x'')$ while holding $M$ fixed at $M_{x'}$, the value that $M$ would attain under  $do(X=x')$. 
\citet{Imai2010a} presented the following identification result under an assumption called sequential ignorability:  
\begin{assumption}[Sequential Ignorability]
\label{SCAS2}
The following two 
conditional independences hold:
$\{Y_{x,m},M_{x'}\}\indep X\text{ and }M_{x'} \indep Y_{x,m}$,
where $\mathfrak{p}(X=x)>0$ and $\mathfrak{p}(M=m|X=x)>0$, for any $m \in \Omega_M$ and $x \in \Omega_X$.
\end{assumption}
\begin{proposition}
\label{prop1}
Under SCM $\mathcal{M}_N$ and Assumption \ref{SCAS2}, the expectation of the counterfactual $\mathbb{E}[Y_{x,M_{x'}}]$ is identifiable by
\begin{equation}
\begin{aligned}
&\mathbb{E}[Y_{x,M_{x'}}]=\int_{\Omega_{M}}\mathbb{E}[Y|X=x,M=m]\mathfrak{p}(M=m|X=x')d{m}
\end{aligned}
\end{equation}
for any $x, x' \in \Omega_X$.
\end{proposition}
{\bf Local Total Effect.}
Average partial causal effect (APCE) is another useful measure to evaluate the total effect of $X$ on $Y$ for a continuous treatment \citep{Chamberlain1984,Kawakami2023}.
In this paper, we call it the local total effect (LTE), which captures the infinitesimal total effect of $X$ at $X=x$:
\begin{definition}[LTE]
The local total effect (LTE) is defined by
\begin{equation}
\text{LTE}(x;Y)\defeq\mathbb{E}[\partial_xY_{x}]
\end{equation}
for any $x \in \Omega_X$.
\end{definition}
TE can be calculated using LTE by
\begin{equation}
\text{TE}(x'',x';Y)=\int_{x'}^{x''}\text{LTE}(x;Y)dx.
\end{equation}
{\bf Popular Mediation Analysis Models.}
Next, we review two widely used models in the causal mediation analysis literature.

\begin{example}[Linear No-Interaction Model] 
\label{eg-m1}
The following linear SCM, denoted by $\mathcal{M}_1$ and depicted by Figure \ref{DAG1}, is the most widely used specification in mediation analysis. It corresponds to the classical Baron and Kenny approach \citep{Baron1986}:
\begin{equation}\label{eq-m1}
   \mathcal{M}_1: X:=\epsilon_X,\ \  M:=\alpha X+\epsilon_{M},\ \
    Y:=\beta X+\gamma M +\epsilon_Y,
\end{equation}
where the exogenous variables $\epsilon_X$, $\epsilon_M$, and $\epsilon_Y$ are mutually independent and satisfy $\mathbb{E}[\epsilon_X]=\mathbb{E}[\epsilon_{M}]=\mathbb{E}[\epsilon_Y]=0$. 
$\mathcal{M}_1$ is a special case of $\mathcal{M}_{N}$ and the structural equation for $Y$ contains no interaction term between $X$ and $M$, commonly referred to as the no-interaction assumption \citep{Robins2003}. 
In this setting, 
the TE is decomposed into NDE and NIE as follows \citep{Vanderweele2009}:
\begin{equation}
\label{eq15}
\begin{aligned}
&\underbrace{(\beta+\alpha\gamma)(x''-x')}_{\text{TE}(x'',x';Y)}=\underbrace{\beta(x''-x')}_{\text{NDE}(x'',x';Y)}-\underbrace{\alpha\gamma(x'-x'')}_{\text{NIE}(x',x'';Y)}.
\end{aligned}
\end{equation}
And $\text{CDE}(x'',x',m;Y)$ is equal to $\beta(x''-x')$ for any $m \in \Omega_M$ and coincides with $\text{NDE}(x'',x';Y)$.
We have that both $\text{NDE}(x'',x';Y)$ and $\text{NIE}(x'',x';Y)$ satisfy skew-symmetry and additivity.
Further, the TE can be additively decomposed into NDE and NIE as follows:
\begin{equation}\label{eq-additive}
\text{TE}(x'',x';Y)=\text{NDE}(x'',x';Y)+\text{NIE}(x'',x';Y).
\end{equation}
\end{example}


\begin{example}[Linear Interaction Model]
 \label{eg-m2}   
Next, we show another widely used model, a linear SCM with an interaction term $XM$, denoted by $\mathcal{M}_2$  \citep{Vanderweele2009}:
\begin{equation}\label{eq-m2}
\begin{gathered}
   \mathcal{M}_2: X:=\epsilon_X,\ \  M:=\alpha X+\epsilon_{M},\ \  Y:=\beta X+\gamma M+ \delta XM+\epsilon_Y,
\end{gathered}
\end{equation}
where $\epsilon_X$, $\epsilon_M$, and $\epsilon_Y$ are mutually independent exogenous variables, and $\mathbb{E}[\epsilon_X]=\mathbb{E}[\epsilon_{M}]=\mathbb{E}[\epsilon_Y]=0$.
The explicit forms of NDE and NIE 
under $\mathcal{M}_2$ have been given by \citet{Vanderweele2009}, 
and the decomposition of TE into NDE and NIE is given by
\begin{equation}
\label{EXA1}
\begin{aligned}
&\underbrace{(\beta+\alpha\gamma)(x''-x')+\alpha\delta (x''^2-x'^2)}_{\text{TE}(x'',x';Y)}
=\underbrace{(\beta+\alpha\delta {x'})(x''-x')}_{\text{NDE}(x'',x';Y)}-\underbrace{(\alpha\gamma+\alpha\delta {x''})(x'-x'')}_{\text{NIE}(x',x'';Y)}.
\end{aligned}
\end{equation}
$\text{CDE}(x'',x',m;Y)$ is given by $\beta(x''-x')+\delta(x''-x')m$, which depends on the value of $m$ and differs from NDE.
In SCM $\mathcal{M}_2$, neither NDE nor NIE satisfies the skew-symmetry or additivity, and the additive decomposition in Eq.~(\ref{eq-additive}) does not hold.
\end{example}

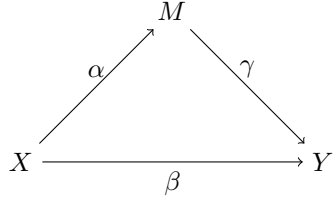
\begin{figure}[tb]
\centering
\scalebox{1}{
\begin{tikzpicture}
\node (z) at (0,2) {$M$};
\node (y) at (2,0) {$Y$};
\node (x) at (-2,0) {$X$};

  \path (x) edge[->] node[midway, above] {$\alpha$} (z);
  \path (z) edge[->] node[midway, above] {$\gamma$} (y);
  \path (x) edge[->] node[midway, below] {$\beta$}  (y);
\end{tikzpicture}
}
\vspace{-0cm}
\caption{The causal graph representing SCM $\mathcal{M}_{1}$ (Baron and Kenny's model).}
\label{DAG1}
\end{figure}

\section{Mediation Paradox \label{sec-paradox}}

In this section, we examine how the lack of skew-symmetry and additivity in the NDE and NIE can lead to paradoxical interpretations in mediation analysis.

As an illustrative example, we consider  assessing the effect of cigarette smoking on heart disease mediated through hypertension \citep{Hernan2023}.
We let the amount of cigarettes per day be the treatment ($X$), a person's blood pressure be the mediator ($M$), and the degree of heart disease be the outcome ($Y$).
Researchers try to answer the following question:
\begin{center}
{\it
``Through which mechanism does smoking increase a person’s risk of heart disease?"
}
\end{center}
Specifically, researchers may want to know whether smoking cigarettes raises the risk of heart disease by raising a person's blood pressure or whether smoking cigarettes raises the risk of heart disease directly.
Causal mediation analyses have empirically answered such  questions by decomposing the total effect into the direct and indirect effects.

First, we illustrate how the lack of symmetry in NDE and NIE provokes paradoxes.
We consider $\mathcal{M}_2$ in Example~\ref{eg-m2}. 
Suppose $\beta=0$, $\alpha=1$, $\gamma=0$, and  $\delta=1$. Then 
TE can be decomposed by Eq.~(\ref{EXA1}) in two ways using NDE and NIE, obtained by flipping the treatment values $1$ and $0$, as shown below:
\begin{equation}\label{eq-m2nde}
\begin{aligned}
&\underbrace{1}_{\text{TE}(1,0;Y)}=\underbrace{0}_{\text{NDE}(1,0;Y)}-\underbrace{(-1)}_{\text{NIE}(0,1;Y)},
\underbrace{-1}_{\text{TE}(0,1;Y)}=\underbrace{-1}_{\text{NDE}(0,1;Y)}-\underbrace{0}_{\text{NIE}(1,0;Y)}.
\end{aligned}
\end{equation}
This represents a paradoxical situation because the decomposition of $\text{TE}(1,0;Y)$ implies there is no direct effect, while the decomposition of $\text{TE}(0,1;Y)$ implies there is no indirect effect.
Just by flipping the values of treatments $1$ and $0$, the result of causal mediation analysis is completely overturned.
This could be unacceptable for most researchers and provoke severe problems in the interpretation of the results of the mediation analysis.

For the cigarette example, the above results state that 
\begin{center}
{\it
``The effect of increasing cigarettes on heart disease is completely via the effect of increasing blood pressure; 
on the other hand, the effect of reducing cigarettes on heart disease is not via the effect of reducing blood pressure."
}
\end{center}
This seems to be a contradiction.

In general, because NDE and NIE do not satisfy skew-symmetry, the proportions of the direct and indirect effects in the total effect—NIE/TE is called the ``proportion mediated'' in \citep{VanderWeele2013b}—will change when the reference point is flipped, that is, 
\begin{align}
\frac{\text{NDE}(x'',x';Y)}{\text{TE}(x'',x';Y)}&\neq\frac{\text{NDE}(x',x'';Y)}{\text{TE}(x',x'';Y)},\
\frac{\text{NIE}(x'',x';Y)}{\text{TE}(x'',x';Y)}\neq\frac{\text{NIE}(x',x'';Y)}{\text{TE}(x',x'';Y)}
\end{align}
This lack of invariance may lead to contradictory interpretations in mediation analysis.


Next, we explain how the lack of additivity also induces paradoxical situations.
\begin{example} \label{eg-m3}
We consider a  linear SCM with an interaction term $X^2M$, $\mathcal{M}_3$:
\begin{equation}\label{eq-m3}
\begin{gathered}
     \mathcal{M}_3:   X:=\epsilon_X,\ \  M:=\alpha X+\epsilon_{M},\ \  Y:=\beta X+\gamma M+ \delta X^2 M+\epsilon_Y,
    \end{gathered}
\end{equation}
where $\epsilon_X$, $\epsilon_{M}$, and $\epsilon_Y$ are mutually independent exogenous variables, and $\mathbb{E}[\epsilon_X]=\mathbb{E}[\epsilon_{M}]=\mathbb{E}[\epsilon_Y]=0$. 
We have the following decomposition of TE by NDE and NIE:
\begin{equation}\label{eq-m3nde}
\begin{aligned}
&\underbrace{\{(\beta+\alpha\gamma)(x''-x')+\alpha\delta (x''^3-x'^3)\}}_{\text{TE}(x'',x';Y)}\\
&=\underbrace{\{\beta(x''-x')+\alpha\delta (x''^2-x'^2)x'\}}_{\text{NDE}(x'',x';Y)}-\underbrace{\{\alpha\gamma(x'-x'')+\alpha\delta x''^2(x'-x'')\}}_{\text{NIE}(x',x'';Y)}.
\end{aligned}
\end{equation}
We have that neither NDE nor NIE satisfies the skew-symmetry or additivity.     

Now supposing $\beta=0$, $\alpha=1$, $\gamma=0$, and  $\delta=1$, then $\text{TE}(1,-1;Y)$ 
is decomposed into NDE and NIE as follows: 
\begin{equation}
\label{eq-m3e1}
\begin{aligned}
\underbrace{2}_{\text{TE}(1,-1;Y)}=\underbrace{0}_{\text{NDE}(1,-1;Y)}-\underbrace{(-2)}_{\text{NIE}(-1,1;Y)}.
\end{aligned}
\end{equation}
Meanwhile, $\text{TE}(1,0;Y)$ and $\text{TE}(0,-1;Y)$ are decomposed by
\begin{equation}
\begin{aligned}
\label{eq-m3e2}
&\underbrace{1}_{\text{TE}(1,0;Y)}=\underbrace{0}_{\text{NDE}(1,0;Y)}-\underbrace{(-1)}_{\text{NIE}(0,1;Y)},
\underbrace{1}_{\text{TE}(0,-1;Y)}=\underbrace{1}_{\text{NDE}(0,-1;Y)}-\underbrace{0}_{\text{NIE}(-1,0;Y)}.
\end{aligned}
\end{equation}
The above results appear paradoxical: $\text{TE}(1,0;Y)$ is decomposed entirely into indirect effects, and $\text{TE}(0,-1;Y)$ is decomposed entirely into direct effects; whereas their sum, $\text{TE}(1,-1;Y)$, is decomposed entirely into indirect effects. This inconsistency makes it difficult to interpret the direct and indirect pathways across treatment contrasts in a coherent manner. 
For the cigarette example, letting $X=0$ be a reference of the amount of cigarettes, the above results state that  
\begin{center}
{\it
``The effect of increasing cigarette consumption by one unit above the reference level is entirely mediated through increased blood pressure. In contrast, the effect of increasing cigarette consumption by one unit from below the reference level to the reference level is entirely direct and not mediated through blood pressure. However, the effect of increasing cigarette consumption by two units across the reference level is again entirely mediated through increased blood pressure."
}
\end{center}
This interpretation also appears contradictory, as the mediation mechanism changes dramatically depending on how the treatment contrast is defined.
\end{example}

These paradoxes could be unacceptable for most researchers, and it would  be desirable to avoid such paradoxical decompositions of TE into NDE and NIE.
In this paper, we propose new measures of direct and indirect effects that avoid these paradoxes.



\section{Direct and Indirect Effects for Continuous Treatment}
\label{SEC3}

In this section, we provide a new decomposition of TE based on the decomposition of LTE.
It preserves both skew-symmetry and additivity under an assumption called differentiability and continuity.

\subsection{Decomposition of LTE}

We explain the nonparametric decomposition of LTE in preparation for defining a new decomposition of TE that satisfies both skew-symmetry and additivity.
We first make the following assumption:
\begin{assumption}[Differentiability and Continuity]
\label{SCAS1}
For any $x,x^* \in \Omega_X$, $\mathbb{E}[Y_{x,{M}_{x^*}}]$ is partially differentiable w.r.t. $x$ and $x^*$ respectively, and $\partial_{x^*} \mathbb{E}[Y_{x,{M}_{x^*}}]$ and $\partial_{x^*} \mathbb{E}[Y_{x^*,{M}_x}]$ are continuous w.r.t. $x$.
\end{assumption}
\noindent This assumption is reasonable in settings where the treatment, outcome, and mediator are all continuous. In particular, the three example SCMs, $\mathcal{M}_1$, $\mathcal{M}_2$, and $\mathcal{M}_3$, all satisfy this assumption.

We then define 
local direct and indirect effects w.r.t. LTE:\footnote{
The LNDE and (negative of) LNIE are referred to as the instantaneous natural direct and indirect effects in \citep{Knafl2017}.
They studied the estimation of these effects under specific polynomial models.
In this paper, we formally state the assumptions required to define these effects and discuss their nonparametric identification.
} 
\begin{definition}[LNIE, LNDE, and LCDE]\label{def-pnie}
Under SCM $\mathcal{M}_N$ and Assumption \ref{SCAS1}, we define 
\begin{enumerate}
\setlength{\itemsep}{2pt}
\setlength{\parskip}{4pt}
\item Local natural direct effect (LNDE):
$\text{LNDE}(x;Y)\defeq\partial_{x^*} \mathbb{E}[Y_{x^*,{M}_{x}}]|_{x^*=x}$
\item Local natural indirect effect (LNIE):
$\text{LNIE}(x;Y)\defeq\partial_{x^*} \mathbb{E}[Y_{x,{M}_{x^*}}]|_{x^*=x}$
\item Local controlled direct effect (LCDE):
$\text{LCDE}(x,m;Y)\defeq\partial_{x} \mathbb{E}[Y_{x,m}]$
\end{enumerate}
for any $x \in \Omega_X$ and $m \in \Omega_M$. 
\end{definition}
\noindent In Def.~\ref{def-pnie}$, \partial_{x^*} \mathbb{E}[Y_{x^*,{M}_{x}}]$ represents
$\lim_{h \rightarrow 0}\frac{\mathbb{E}[Y_{x^*+h,{M}_{x}}]-\mathbb{E}[Y_{x^*,{M}_{x}}]}{h}$ and $\partial_{x^*} \mathbb{E}[Y_{x,{M}_{x^*}}]$ represents $\lim_{h \rightarrow 0}\frac{\mathbb{E}[Y_{x,{M}_{x^*+h}}]-\mathbb{E}[Y_{x,{M}_{x^*}}]}{h}$.
LNDE is the slope of the expectation of the counterfactual $Y_{x^*,{M}_{x}}$ w.r.t. $x^*$ at $x^*=x$, 
which represents the infinitesimal effect of the treatment at $X=x$ while holding the mediator $M$ at its natural level under $X=x$.  
LNIE is the slope of the expectation of the counterfactual $Y_{x,{M}_{x^*}}$ w.r.t. $x^*$ at $x^*=x$, which represents the influence of the treatment at $X=x$ through $M_{x}$, while keeping the treatment itself fixed to $X=x$.
LCDE is the slope of the expectation of the counterfactual $Y_{x,m}$ w.r.t. $x$.
It depends on $m$ and does not have a parallel indirect effect.

We have the following important relationship that the LTE is  additively decomposed into LNDE and LNIE:
\begin{restatable}{lemma}{Lemmaone}[Decomposition of LTE by LNDE and LNIE]
\label{lemma1}
Under SCM $\mathcal{M}_{N}$ and Assumption \ref{SCAS1}, 
\begin{equation}
\text{LTE}(x;Y)=\text{LNDE}(x;Y)+\text{LNIE}(x;Y)    
\end{equation}
for any $x \in \Omega_X$.
\end{restatable}

Next, we study the identifiability of LNDE and LNIE. 
They can be identified through identifying $\mathbb{E}[Y_{x,{M}_{x^*}}]$ via Proposition~\ref{prop1}. 
Once $\mathbb{E}[Y_{x,{M}_{x^*}}]$ is identified, LNDE and LNIE are given by  $\partial_{x^*} \mathbb{E}[Y_{x^*,{M}_{x}}]|_{x^*=x}$ and $\partial_{x^*} \mathbb{E}[Y_{x,{M}_{x^*}}]|_{x^*=x}$, respectively.
We next provide explicit expressions to facilitate estimation.
\begin{restatable}{theorem}{Theoremone}[Identification of LNDE and LNIE]
\label{cor3}
Under SCM $\mathcal{M}_N$ and Assumptions \ref{SCAS2} and \ref{SCAS1}, 
if the following regularity condition holds: either
 (i) $\mathbb{E}[Y|X=x,M=m]$ is bounded and $M$ has bounded support; or
(ii) $\mathbb{E}[Y|X=x, M=m]$ grows at most polynomially in $m$, and the exogenous error $\varepsilon_M$ in $\mathcal{M}_N$ 
has a continuously differentiable density $\mathfrak{p}$ such that $\mathfrak{p}'(m)$ is integrable and $\mathfrak{p}(m) \to 0$
 as $|m| \to \infty$ faster than any polynomial (e.g., Gaussian),
then LNDE and LNIE are identifiable for any $x \in \Omega_X$ by
\begin{align}
&\text{LNDE}(x;Y)=\int_{\Omega_{M}} \partial_x \mathbb{E}[Y|X=x,M=m]\mathfrak{p}(M=m|X=x)d{m},\\
&\text{LNIE}(x;Y)=\int_{\Omega_M}\partial_{m}\mathbb{E}[Y|X=x,M=m]\partial_x \mathbb{P}(M>m|X=x)dm.
\end{align}
\end{restatable}
\noindent {The regularity condition in Theorem~\ref{cor3} ensures the validity of the integration-by-parts argument used in the proof.}

\subsection{Decomposition of TE via decomposing LTE}
Next, we introduce new measures of direct and indirect effects that satisfy skew-symmetry and additivity through  LNDE and LNIE.
\begin{definition}[CNDE and CNIE]
Under SCM $\mathcal{M}_N$ and Assumption \ref{SCAS1}, 
we define the cumulative natural direct and indirect effects (CNDE and CNIE) over the interval $[x', x'']$ as the integral of the corresponding local effects: 
\begin{equation}
\begin{aligned}
\text{CNDE}(x'',x';Y)&\defeq\int_{x'}^{x''} \text{LNDE}(x;Y) dx, \\ 
\text{CNIE}(x'',x';Y)&\defeq\int_{x'}^{x''}\text{LNIE}(x;Y)dx
\end{aligned}
\end{equation}
for any $x', x'' \in \Omega_X$. 
\end{definition}
CNDE and CNIE capture the accumulated direct and indirect effects along the treatment path from $x'$ to $x''$, respectively.  
Note the distinction from NDE and NIE, which, under Assumption \ref{SCAS1}, can be expressed as 
\begin{align}
\text{NDE}(x'',x';Y)&=\left(\int_{x'}^{x''} \partial_{x^*}\mathbb{E}[Y_{{x^*},M_{x'}}]d{x^*}\right), \\
\text{NIE}(x'',x';Y)&=\left(\int_{x'}^{x''} \partial_{x^*}\mathbb{E}[Y_{x',M_{x^*}}]dx^*\right).
\end{align}
The key distinction lies in how the effects are accumulated along the interval. The CNDE and CNIE integrate \emph{local infinitesimal effects} evaluated along the treatment path, so that both the treatment level and the reference mediator value evolve with $x$. This yields a pathwise decomposition that aggregates infinitesimal direct and indirect effects as the treatment moves from $x'$ to $x''$.
In contrast, the NDE and NIE are constructed by integrating derivatives while holding one argument fixed at an endpoint (either $x'$ or $x''$), so the accumulation is taken \emph{relative to a fixed reference regime} rather than along the evolving path.

CNDE and CNIE satisfy both skew-symmetry and additivity: 
\begin{restatable}{theorem}{Theoremtwo}[Skew-symmetry and additivity of CNDE and CNIE]
\label{theo2}
Under SCM $\mathcal{M}_{N}$ and Assumption \ref{SCAS1}, CNDE and CNIE satisfy both skew-symmetry and additivity.
\end{restatable}
We can easily extend additivity for arbitrary numbers of treatments, e.g., 
\begin{equation}
\begin{aligned}
\text{CNDE}(x'''',x';Y)&=\text{CNDE}(x'''',x''';Y)+\text{CNDE}(x''',x'';Y)+\text{CNDE}(x'',x';Y)\\
\text{CNIE}(x'''',x';Y)&=\text{CNIE}(x'''',x''';Y)+\text{CNIE}(x''',x'';Y)+\text{CNIE}(x'',x';Y)
\end{aligned}
\end{equation}
for any $x', x'',x''', x'''' \in \Omega_X$. 

Furthermore, we have the following decomposition theorem:
\begin{restatable}{theorem}{Theoremthree}[Decomposition of TE into CNDE and CNIE]
\label{theo3}
Under SCM $\mathcal{M}_{N}$ and Assumption \ref{SCAS1}, 
\begin{align}
\text{TE}(x'',x';Y)
&=\text{CNDE}(x'',x';Y)-\text{CNIE}(x',x'';Y)\\
&=\text{CNDE}(x'',x';Y)+\text{CNIE}(x'',x';Y) \label{eq-decom-c}
\end{align}
for any $x',x'' \in \Omega_X$.
\end{restatable}
\noindent We note that the additive decomposition of TE into NDE and NIE only holds for the linear no-interaction model. In contrast, the additive decomposition of TE into CNDE and CNIE in Eq.~(\ref{eq-decom-c}) holds for any nonparametric model.
In addition, since CNDE and CNIE satisfy skew-symmetry, the proportion of direct
or indirect effect in the total effect does not change when $x''$ and $x'$ are flipped
\begin{align}
\frac{\text{CNDE}(x'',x';Y)}{\text{TE}(x'',x';Y)}=\frac{\text{CNDE}(x',x'';Y)}{\text{TE}(x',x'';Y)},\ \  
\frac{\text{CNIE}(x'',x';Y)}{\text{TE}(x'',x';Y)}=\frac{\text{CNIE}(x',x'';Y)}{\text{TE}(x',x'';Y)}
\end{align}



Finally, by Theorem \ref{cor3}, we can identify CNDE and CNIE as follows. 
\begin{restatable}{theorem}{Theoremfour}[Identification of CNDE and CNIE]
\label{theo4}
Under SCM $\mathcal{M}_{N}$ and Assumptions \ref{SCAS2} and \ref{SCAS1}, 
for any $x', x'' \in \Omega_X$,
CNDE and CNIE are identifiable by
\begin{equation}
\begin{aligned}
&\text{CNDE}(x'',x';Y)=\int_{x'}^{x''}\int_{\Omega_{M}}\partial_x \mathbb{E}[Y|X=x,M=m]\mathfrak{p}(M=m|X=x)d{m}dx,
\end{aligned}
\end{equation}
\begin{equation}
\begin{aligned}
&\text{CNIE}(x'',x';Y)=\int_{x'}^{x''}\int_{\Omega_M}\partial_{m}\mathbb{E}[Y|X=x,M=m]\partial_x \mathbb{P}(M>m|X=x)dmdx.
\end{aligned}
\end{equation}
\end{restatable}

{\bf Remark: Covariate Adjustment.}
Assumption~\ref{SCAS2} may be overly restrictive for observational studies. In Appendix~\ref{SECB}, we provide identification results that incorporate observed confounders $W$. 




\subsection{Revisiting Popular Mediation Models}

Next, we revisit the three specific models and illustrate how TE is decomposed into CNDE and CNIE.

\begin{customeg}{1'}   
Under the linear SCM $\mathcal{M}_1$ in (\ref{eq-m1}), we have 
\begin{equation}
\begin{aligned}
&\text{LTE}(x;Y)=\beta+\alpha\gamma,\ \ 
\text{LNDE}(x;Y)=\beta,\ \ 
\text{LNIE}(x;Y)=\alpha\gamma.
\end{aligned}
\end{equation} 
TE is decomposed into  CNDE and CNIE as below:
\begin{equation}
\begin{aligned}
&\underbrace{(\beta+\alpha\gamma)(x''-x')}_{\text{TE}(x'',x';Y)}=
\underbrace{\beta(x''-x')}_{\text{CNDE}(x'',x';Y)}-\underbrace{\alpha\gamma(x'-x'')}_{\text{CNIE}(x',x'';Y)}.
\end{aligned}
\end{equation}
We have that CNDE and CNIE coincide with NDE and NIE, respectively.
\end{customeg}

\begin{customeg}{2'}   
Consider the SCM $\mathcal{M}_2$ in (\ref{eq-m2}). 
We have 
\begin{align}\text{LTE}(x;Y) = \beta + \alpha\gamma + 2\alpha\delta x, \ \text{LNDE}(x;Y) = \beta + \alpha\delta x, \ \text{LNIE}(x;Y) = \alpha\gamma + \alpha\delta x.
\end{align}
The decomposition of TE into CNDE and CNIE becomes 
\begin{equation}\label{eq-ex2cnde}
\begin{aligned}
&\underbrace{(\beta+\alpha\gamma)(x''-x')+\alpha\delta (x''^2-x'^2)}_{\text{TE}(x'',x';Y)}\\
&=\underbrace{\{\beta(x''-x')+0.5\alpha\delta (x''^2-x'^2)\}}_{\text{CNDE}(x'',x';Y)}-\underbrace{\{\alpha\gamma(x'-x'') +0.5\alpha\delta (x'^2-x''^2)\}}_{\text{CNIE}(x',x'';Y)}\\
&=\underbrace{\{\beta(x''-x')+0.5\alpha\delta (x''^2-x'^2)\}}_{\text{CNDE}(x'',x';Y)}+\underbrace{\{\alpha\gamma(x''-x') +0.5\alpha\delta (x''^2-x'^2)\}}_{\text{CNIE}(x'',x';Y)}.
\end{aligned}
\end{equation}
This decomposition differs from the corresponding decomposition into NDE and NIE in Eq.~\eqref{EXA1}. 
Concretely, supposing $\beta=0$, $\alpha=1$, $\gamma=0$, and  $\delta=1$, then $\text{TE}(1,0;Y)$ is decomposed by 
\begin{equation}
\begin{aligned}
\underbrace{1}_{\text{TE}(1,0;Y)}=\underbrace{0.5}_{\text{CNDE}(1,0;Y)}-\underbrace{(-0.5)}_{\text{CNIE}(0,1;Y)}=\underbrace{0.5}_{\text{CNDE}(1,0;Y)}+\underbrace{0.5}_{\text{CNIE}(1,0;Y)}.
\end{aligned}
\end{equation} 
Meanwhile, $\text{TE}(0,1;Y)$ is decomposed by
\begin{equation}
\begin{aligned}
&\underbrace{-1}_{\text{TE}(0,1;Y)}=\underbrace{-0.5}_{\text{CNDE}(0,1;Y)}-\underbrace{0.5}_{\text{CNIE}(1,0;Y)}=\underbrace{-0.5}_{\text{CNDE}(0,1;Y)}+\underbrace{(-0.5)}_{\text{CNIE}(0,1;Y)}.
\end{aligned}
\end{equation} 
For the cigarette example, the above results say that
\begin{center}
``{\it Half of the effect of increasing cigarette consumption by one unit from the reference level is via the effect of increasing blood pressure; also, similarly, half of the effect of decreasing cigarette consumption by one unit from the reference level is via decreasing blood pressure.}"
\end{center}
In contrast to the conventional decomposition into NDE and NIE in Eq.~\eqref{eq-m2nde}, the decomposition into CNDE and CNIE does not exhibit  paradoxical phenomena.

\end{customeg}

\begin{customeg}{3'}   

Consider the SCM $\mathcal{M}_3$ in (\ref{eq-m3}). 
We have 
\begin{align}
\text{LTE}(x;Y) = \beta + \alpha\gamma + 3\alpha\delta x^2,\ \text{LNDE}(x;Y) = \beta + 2\alpha\delta x^2,\ \text{LNIE}(x;Y) = \alpha\gamma + \alpha\delta x^2.
\end{align}
The decomposition of TE into CNDE and CNIE becomes
\begin{equation}
\begin{aligned}
&\underbrace{(\beta+\alpha\gamma)(x''-x')+\alpha\delta (x''^3-x'^3)}_{\text{TE}(x'',x';Y)}\\
&=\underbrace{\left\{\beta(x''-x')+\frac{2}{3}\alpha\delta (x''^3-x'^3)\right\}}_{\text{CNDE}(x'',x';Y)}-\underbrace{\left\{\alpha\gamma(x'-x'')+\frac{1}{3}\alpha\delta (x'^3-x''^3)\right\}}_{\text{CNIE}(x',x'';Y)}\\
&=\underbrace{\left\{\beta(x''-x')+\frac{2}{3}\alpha\delta (x''^3-x'^3)\right\}}_{\text{CNDE}(x'',x';Y)}+\underbrace{\left\{\alpha\gamma(x''-x')+\frac{1}{3}\alpha\delta (x''^3-x'^3)\right\}}_{\text{CNIE}(x'',x';Y)}.
\end{aligned}
\end{equation} 
This decomposition differs from the corresponding decomposition into NDE and NIE in Eq.~\eqref{eq-m3nde}. 
Concretely, supposing $\beta=0$, $\alpha=1$, $\gamma=0$, and  $\delta=1$, then 
$\text{TE}(1,-1;Y)$ is decomposed by 
\begin{equation}\label{eq-m33e1}
\begin{aligned}
\underbrace{2}_{\text{TE}(1,-1;Y)}=\underbrace{4/3}_{\text{CNDE}(1,-1;Y)}-\underbrace{(-2/3)}_{\text{CNIE}(-1,1;Y)}=\underbrace{4/3}_{\text{CNDE}(1,-1;Y)}+\underbrace{2/3}_{\text{CNIE}(1,-1;Y)}.
\end{aligned}
\end{equation} 
Meanwhile, $\text{TE}(1,0;Y)$ and $\text{TE}(0,-1;Y)$ are decomposed by
\begin{equation}\label{eq-m33e2}
\begin{aligned}
&\underbrace{1}_{\text{TE}(1,0;Y)}=\underbrace{{2}/{3}}_{\text{CNDE}(1,0;Y)}-\underbrace{\left(-{1}/{3}\right)}_{\text{CNIE}(0,1;Y)}=\underbrace{{2}/{3}}_{\text{CNDE}(1,0;Y)}+\underbrace{{1}/{3}}_{\text{CNIE}(1,0;Y)},\\ 
&\underbrace{1}_{\text{TE}(0,-1;Y)}=\underbrace{{2}/{3}}_{\text{CNDE}(0,-1;Y)}-\underbrace{\left(-{1}/{3}\right)}_{\text{CNIE}(-1,0;Y)}=\underbrace{{2}/{3}}_{\text{CNDE}(0,-1;Y)}+\underbrace{{1}/{3}}_{\text{CNIE}(0,-1;Y)}.
\end{aligned}
\end{equation} 
In contrast to the decomposition into NDE and NIE in Eqs.~\eqref{eq-m3e1} and~\eqref{eq-m3e2}, the decomposition into CNDE and CNIE does not exhibit  paradoxical phenomena as they satisfy additivity: 
\begin{equation}
\begin{aligned}
&\text{CNDE}(1,-1;Y)=\text{CNDE}(1,0;Y)+\text{CNDE}(0,-1;Y),
\end{aligned}
\end{equation}
\begin{equation}
\begin{aligned}
&\text{CNIE}(1,-1;Y)=\text{CNIE}(1,0;Y)+\text{CNIE}(0,-1;Y).
\end{aligned}
\end{equation}
For the cigarette example, the decompositions in \eqref{eq-m33e1} and \eqref{eq-m33e2} imply that
\begin{center}
``One-third of the effect of each of the three treatment contrasts, an increase of one unit above the reference level, an increase of one unit from below the reference level to the reference level, and an increase of two units across the reference level, is mediated through increased blood pressure, while the remaining two-thirds operates directly.''
\end{center}
\end{customeg}

We observe that our proposed decomposition of TE into CNDE and CNIE avoids the paradoxes exhibited by the decomposition into NDE and NIE discussed in Section~\ref{sec-paradox}. 



\section{Direct and Indirect Effects for Ordinal Treatment}
When the treatment $X$ is a discrete variable, our definitions of CNDE and CNIE  are not applicable  because the derivative is not defined in such cases. The natural counterparts of CNDE and CNIE can be defined using finite differences rather than integrals. 

\subsection{CNDE and CNIE for ordinal treatment}

If $X$ is ordinal taking values in $\{x_0, x_1, \dots, x_K\}$, the analogue of CNDE/CNIE can be expressed as sums of local increments along a sequence connecting $x_i$ to $x_j$. We define the counterparts of CNDE and CNIE for the ordinal $X$, denoted by CNDE-O and CNIE-O, as follows:
\begin{definition}[CNDE and CNIE for ordinal $X$]
Let $X$ be an ordinal taking values in $\Omega_X=\{x_0, x_1, \dots, x_K\}$. 
We define CNDE and CNIE for the ordinal $X$, CNDE-O and CNIE-O, as follows: 
\begin{align}
 \text{CNDE-O}(x_j,x_i;Y) &\defeq 
 \begin{cases}
   \sum_{k=i}^{j-1} \text{NDE}(x_{k+1},x_k;Y) & \text{if } i<j\\  
   \sum_{k=j}^{i-1} \text{NDE}(x_{k},x_{k+1};Y) & \text{if } i>j\\  
 \end{cases},\\
 \text{CNIE-O}(x_j,x_i;Y) &\defeq 
 \begin{cases}
   \sum_{k=i}^{j-1} \text{NIE}(x_{k+1},x_k;Y) & \text{if } i<j\\  
   \sum_{k=j}^{i-1} \text{NIE}(x_{k},x_{k+1};Y) & \text{if } i>j\\  
 \end{cases}, 
\end{align}
for any $x_i, x_j \in \Omega_X$.
\end{definition}
Thus, the discrete analogue replaces integration of local derivatives with summation of finite-step contrasts. In the case of $|i-j|=1$, this reduces to a single-step decomposition, making CNDE-O/CNIE-O coincide with NDE/NIE.

CNDE-O and CNIE-O do not satisfy skew-symmetry because NDE and NIE do not. CNDE-O and CNIE-O satisfy a restricted form of additivity as stated in the following: 

\begin{restatable}{theorem}{Theoremfive}[Additivity of CNDE-O and CNIE-O]
Under SCM $\mathcal{M}_{N}$, 
CNDE-O and CNIE-O satisfy: 
if $x' < x'' < x'''$ or $x' > x'' > x'''$, then
\begin{align}
 \text{CNDE-O}(x''',x';Y)&= \text{CNDE-O}(x''',x'';Y)+ \text{CNDE-O}(x'',x';Y),\\
 \text{CNIE-O}(x''',x';Y)&= \text{CNIE-O}(x''',x'';Y)+ \text{CNIE-O}(x'',x';Y).
\end{align}
\end{restatable}

We have the following decomposition theorem:
\begin{restatable}{theorem}{Theoremsix}[Decomposition of TE into CNDE-O and CNIE-O]
Under SCM $\mathcal{M}_{N}$, 
\begin{align}
\text{TE}(x'',x';Y)
&=\text{CNDE-O}(x'',x';Y)-\text{CNIE-O}(x',x'';Y).
\end{align}
\end{restatable}


Next, we revisit the three models $\mathcal{M}_1$, $\mathcal{M}_2$, and $\mathcal{M}_3$.

\begin{customeg}{1''}  
Under $\mathcal{M}_1$ in (\ref{eq-m1}), for $x'=x_i$ and $x''=x_j$ with $i<j$, we have $\text{NDE}(x_{k+1},x_k;Y) = \beta(x_{k+1}-x_k)$ and $\text{NIE}(x_{k+1},x_k;Y) = \alpha\gamma(x_{k+1}-x_k)$. Then CNDE-O and CNIE-O are given by
\begin{equation}
\begin{aligned}
\text{CNDE-O}(x'',x';Y) &= \sum_{k=i}^{j-1}\beta(x_{k+1}-x_k) = \beta(x_j-x_i) = \beta(x''-x'),\\
\text{CNIE-O}(x',x'';Y) &= \sum_{k=i}^{j-1}\alpha\gamma(x_k-x_{k+1}) = \alpha\gamma(x_i-x_j) = \alpha\gamma(x'-x'').
\end{aligned}
\end{equation}
They coincide with NDE/CNDE and NIE/CNIE respectively. 
\end{customeg}

\begin{customeg}{2''}
Under $\mathcal{M}_2$ in (\ref{eq-m2}), 
we have 
$\text{NDE}(x_{k+1},x_k;Y)=(\beta+\alpha\delta x_k)(x_{k+1}-x_k)$ and 
$\text{NIE}(x_k,x_{k+1};Y)=(\alpha\gamma+\alpha\delta x_{k+1})(x_k-x_{k+1})$. 
Then for $i<j$, CNDE-O and CNIE-O are given by
\begin{equation}
\begin{aligned}
&\text{CNDE-O}(x_j,x_i;Y) = \sum_{k=i}^{j-1}(\beta+\alpha\delta x_k)(x_{k+1}-x_k),\\
&\text{CNIE-O}(x_i,x_j;Y)
= \sum_{k=i}^{j-1}(\alpha\gamma+\alpha\delta x_{k+1})(x_k-x_{k+1}).
\end{aligned}
\end{equation}
Supposing $\beta=0$, $\alpha=1$, $\gamma=0$, $\delta=1$, and $X\in\{-1,0,1\}$, then
\begin{equation}
\begin{aligned}
&\underbrace{0}_{\text{TE}(1,-1;Y)}=\underbrace{-1}_{\text{CNDE-O}(1,-1;Y)}-\underbrace{(-1)}_{\text{CNIE-O}(-1,1;Y)},\\
&\underbrace{1}_{\text{TE}(1,0;Y)}=\underbrace{0}_{\text{CNDE-O}(1,0;Y)}-\underbrace{(-1)}_{\text{CNIE-O}(0,1;Y)},\\
&\underbrace{-1}_{\text{TE}(0,-1;Y)}=\underbrace{-1}_{\text{CNDE-O}(0,-1;Y)}-\underbrace{0}_{\text{CNIE-O}(-1,0;Y)}.
\end{aligned}
\end{equation}
We can verify that the additivity holds.
\end{customeg}

\begin{customeg}{3''}  
Under $\mathcal{M}_3$ in (\ref{eq-m3}), 
we have 
$\text{NDE}(x_{k+1},x_k;Y)=\{\beta+\alpha\delta(x_{k+1}+x_k)x_k\}(x_{k+1}-x_k)$ and 
$\text{NIE}(x_k,x_{k+1};Y)=\{\alpha\gamma+\alpha\delta x_{k+1}^2\}(x_k-x_{k+1})$. 
Then for $i<j$, CNDE-O and CNIE-O are given by
\begin{equation}
\begin{aligned}
\text{CNDE-O}(x_j,x_i;Y)
&= \sum_{k=i}^{j-1}\{\beta+\alpha\delta (x_{k+1}+x_k)x_k\}(x_{k+1}-x_k),\\
\text{CNIE-O}(x_i,x_j;Y)
&= \sum_{k=i}^{j-1}\{\alpha\gamma+\alpha\delta x_{k+1}^2\}(x_k-x_{k+1}).
\end{aligned}
\end{equation}
Suppose $\beta=0$, $\alpha=1$, $\gamma=0$, $\delta=1$, and $X\in\{-1,0,1\}$. Then
\begin{equation}
\begin{aligned}
&\underbrace{2}_{\text{TE}(1,-1;Y)}=\underbrace{1}_{\text{CNDE-O}(1,-1;Y)}-\underbrace{(-1)}_{\text{CNIE-O}(-1,1;Y)},\\
&\underbrace{1}_{\text{TE}(1,0;Y)}=\underbrace{0}_{\text{CNDE-O}(1,0;Y)}-\underbrace{(-1)}_{\text{CNIE-O}(0,1;Y)},\\
&\underbrace{1}_{\text{TE}(0,-1;Y)}=\underbrace{1}_{\text{CNDE-O}(0,-1;Y)}-\underbrace{0}_{\text{CNIE-O}(-1,0;Y)}.
\end{aligned}
\end{equation}
We can verify that the additivity holds.
\end{customeg}


\subsection{Skew-Symmetric NDE and NIE}

To construct skew-symmetric counterparts of CNDE and CNIE,  
we first define a skew-symmetric version of NDE and NIE, denoted by S-NDE and S-NIE, using the average of the forward and backward endpoint evaluations.
\begin{definition}[S-NDE and S-NIE]
We define a skew-symmetric version of NDE and NIE (S-NDE and S-NIE) as follows:
\begin{align}
 \text{S-NDE}(x'',x';Y) 
  &\defeq \frac{1}{2}\Big(\mathbb{E}[Y_{x'',M_{x'}}] - \mathbb{E}[Y_{x'}]  + \mathbb{E}[Y_{x''}] - \mathbb{E}[Y_{x',M_{x''}}]\Big)\\
&= \frac{1}{2}\Big(\text{NDE}(x'',x';Y) - \text{NDE}(x',x'';Y)\Big).
\end{align} 
\begin{align}
 \text{S-NIE}(x'',x';Y)
 &\defeq \frac{1}{2}\Big(\mathbb{E}[Y_{x',M_{x''} }] - \mathbb{E}[Y_{x'}] + \mathbb{E}[Y_{x''}] - \mathbb{E}[Y_{x'',M_{x'}}]    \Big)\\
 &=\frac{1}{2}\Big(\text{NIE}(x'',x';Y) - \text{NIE}(x',x'';Y)\Big).
\end{align} 
\end{definition}
S-NDE and S-NIE  satisfy the skew-symmetry. However, they do not satisfy the additivity. 
Further, we have the following decomposition theorem: 
\begin{restatable}{theorem}{Theoremseven}[Decomposition of TE into S-NDE and S-NIE]
Under SCM $\mathcal{M}_{N}$, 
\begin{align}
\text{TE}(x'',x';Y)
&=\text{S-NDE}(x'',x';Y)-\text{S-NIE}(x',x'';Y)\\
&=\text{S-NDE}(x'',x';Y)+\text{S-NIE}(x'',x';Y)
\end{align}
for any $x',x'' \in \Omega_X$.
\end{restatable}


Next, we revisit the three models $\mathcal{M}_1$, $\mathcal{M}_2$, and $\mathcal{M}_3$.

\begin{customeg}{1'''}
Under $\mathcal{M}_1$ in (\ref{eq-m1}), we have $\text{NDE}(x'',x';Y)=\beta(x''-x')$ and $\text{NIE}(x'',x';Y)=\alpha\gamma(x''-x')$. Then
\begin{equation}
\begin{aligned}
\text{S-NDE}(x'',x';Y) &= \frac{1}{2}\Big(\beta(x''-x')-\beta(x'-x'')\Big) = \beta(x''-x'),\\
\text{S-NIE}(x',x'';Y) &= \frac{1}{2}\Big(\alpha\gamma(x'-x'')-\alpha\gamma(x''-x')\Big) = \alpha\gamma(x'-x'').
\end{aligned}
\end{equation}
They coincide with NDE and NIE respectively. 
\end{customeg}

\begin{customeg}{2'''}
Under $\mathcal{M}_2$ in (\ref{eq-m2}), we have $\text{NDE}(x'',x';Y)=(\beta+\alpha\delta x')(x''-x')$ and $\text{NIE}(x'',x';Y)=(\alpha\gamma+\alpha\delta x')(x''-x')$. Then
\begin{equation}\label{eq-m2snde}
\begin{aligned}
\text{S-NDE}(x'',x';Y) 
&= \Big(\beta+\frac{\alpha\delta(x'+x'')}{2}\Big)(x''-x'),\\
\text{S-NIE}(x',x'';Y) 
&= \Big(\alpha\gamma+\frac{\alpha\delta(x''+x')}{2}\Big)(x'-x'').
\end{aligned}
\end{equation}
Interestingly, the above S-NDE and S-NIE coinside with CNDE and CNIE in Eq.~\eqref{eq-ex2cnde}, respectively. 
Supposing $\beta=0$, $\alpha=1$, $\gamma=0$, $\delta=1$, then $\text{TE}(1,0;Y)$ is decomposed by
\begin{equation}
\begin{aligned}
&\underbrace{1}_{\text{TE}(1,0;Y)}=\underbrace{0.5}_{\text{S-NDE}(1,0;Y)}-\underbrace{(-0.5)}_{\text{S-NIE}(0,1;Y)}=\underbrace{0.5}_{\text{S-NDE}(1,0;Y)}+\underbrace{0.5}_{\text{S-NIE}(1,0;Y)}.
\end{aligned}
\end{equation}
\end{customeg}

\begin{customeg}{3'''}
Under $\mathcal{M}_3$ in (\ref{eq-m3}), we have $\text{NDE}(x'',x';Y)=\beta(x''-x')+\alpha\delta(x''^2-x'^2)x'$ and $\text{NIE}(x'',x';Y)=(\alpha\gamma+\alpha\delta x'^2)(x''-x')$. Then
\begin{equation}\label{eq-ex3snde3}
\begin{aligned}
\text{S-NDE}(x'',x';Y) &= \beta(x''-x')+\frac{\alpha\delta}{2}(x''^2-x'^2)(x''+x'),\\
\text{S-NIE}(x',x'';Y) &=\alpha\gamma(x'-x'') +\frac{\alpha\delta}{2}(x'^2+x''^2)(x'-x'').
\end{aligned}
\end{equation}
Supposing $\beta=0$, $\alpha=1$, $\gamma=0$, $\delta=1$, then $\text{TE}(1,0;Y)$ is decomposed by
\begin{equation}
\begin{aligned}
&\underbrace{1}_{\text{TE}(1,0;Y)}=\underbrace{0.5}_{\text{S-NDE}(1,0;Y)}-\underbrace{(-0.5)}_{\text{S-NIE}(0,1;Y)}=\underbrace{0.5}_{\text{S-NDE}(1,0;Y)}+\underbrace{0.5}_{\text{S-NIE}(1,0;Y)}.
\end{aligned}
\end{equation}
\end{customeg}

\subsection{Skew-symmetric CNDE and CNIE for ordinal treatment}
Finally, we define skew-symmetric counterparts of CNDE and CNIE  for the ordinal $X$. 
\begin{definition}[S-CNDE-O and S-CNIE-O]
Let $X$ be an ordinal taking values in $\Omega_X=\{x_0, x_1, \dots, x_K\}$. We define the skew-symmetric CNDE and CNIE for ordinal $X$, S-CNDE-O and S-CNIE-O, as follows:
\begin{align}
 \text{S-CNDE-O}(x_j,x_i;Y) &\defeq 
 \begin{cases}
   \sum_{k=i}^{j-1} \text{S-NDE}(x_{k+1},x_k;Y) & \text{if } i<j\\  
   \sum_{k=j}^{i-1} \text{S-NDE}(x_{k},x_{k+1};Y) & \text{if } i>j\\  
 \end{cases}\label{eq-defscndeo}\\
 \text{S-CNIE-O}(x_j,x_i;Y) &\defeq 
 \begin{cases}
   \sum_{k=i}^{j-1} \text{S-NIE}(x_{k+1},x_k;Y) & \text{if } i<j\\  
   \sum_{k=j}^{i-1} \text{S-NIE}(x_{k},x_{k+1};Y) & \text{if } i>j\\  
 \end{cases} \label{eq-defscnieo}
\end{align}
\end{definition}

S-CNDE-O and S-CNIE-O satisfy the skew-symmetry and a restricted form of additivity:
\begin{restatable}{theorem}{Theoremeight}[Skew-symmetry and additivity of S-CNDE-O and S-CNIE-O]
Under SCM $\mathcal{M}_{N}$,
(i) S-CNDE-O and S-CNIE-O satisfy skew-symmetry, and (ii) S-CNDE-O and S-CNIE-O satisfy: if $x' < x'' < x'''$ or $x' > x'' > x'''$, then
\begin{align}
 \text{S-CNDE-O}(x''',x';Y)&= \text{S-CNDE-O}(x''',x'';Y)+ \text{S-CNDE-O}(x'',x';Y),\\
 \text{S-CNIE-O}(x''',x';Y)&= \text{S-CNIE-O}(x''',x'';Y)+ \text{S-CNIE-O}(x'',x';Y).
\end{align}
\end{restatable}

We have the following decomposition theorem:
\begin{restatable}{theorem}{Theoremnine}[Decomposition of TE into S-CNDE-O and S-CNIE-O]
Under SCM $\mathcal{M}_{N}$, 
\begin{align}
\text{TE}(x'',x';Y)
&=\text{S-CNDE-O}(x'',x';Y)-\text{S-CNIE-O}(x',x'';Y)\\
&=\text{S-CNDE-O}(x'',x';Y)+ \text{S-CNIE-O}(x'',x';Y).
\end{align}
\end{restatable}


Next, we revisit the three models $\mathcal{M}_1$, $\mathcal{M}_2$, and $\mathcal{M}_3$.

\begin{customeg}{1''''}
Under $\mathcal{M}_1$  in (\ref{eq-m1}), we have S-CNDE-O = S-NDE= CNDE-O = NDE and S-CNIE-O = S-NIE= CNIE-O = NIE. 
\end{customeg}

\begin{customeg}{2''''}
Under $\mathcal{M}_2$ in (\ref{eq-m2}), plugging Eq.~\eqref{eq-m2snde} into \eqref{eq-defscndeo} and \eqref{eq-defscnieo} respectively, we obtain S-CNDE-O = S-NDE, which coincides with CNDE, and S-CNIE-O = S-NIE,  which coincides with CNIE.

\end{customeg}

\begin{customeg}{3''''}
Under $\mathcal{M}_3$ in (\ref{eq-m3}), by Eq.~\eqref{eq-ex3snde3}, 
For $i<j$, S-CNDE-O and S-CNIE-O are given by
\begin{equation}
\begin{aligned}
\text{S-CNDE-O}(x_j,x_i;Y)&=\sum_{k=i}^{j-1}
\left[\beta
+\frac{\alpha\delta}{2}(x_{k+1}+x_k)^2
\right](x_{k+1}-x_k),\\
\text{S-CNIE-O}(x_j,x_i;Y)&=\sum_{k=i}^{j-1}
\left[\alpha\gamma+\frac{\alpha\delta}{2}(x_{k+1}^2+x_k^2)\right](x_{k+1}-x_k).
\end{aligned}
\end{equation}
Suppose $\beta=0$, $\alpha=1$, $\gamma=0$, $\delta=1$, and $X\in\{-1,0,1\}$. Then
\begin{equation}
\begin{aligned}
&\underbrace{2}_{\text{TE}(1,-1;Y)}=\underbrace{1}_{\text{S-CNDE-O}(1,-1;Y)}-\underbrace{(-1)}_{\text{S-CNIE-O}(-1,1;Y)}
=\underbrace{1}_{\text{S-CNDE-O}(1,-1;Y)}+\underbrace{1}_{\text{S-CNIE-O}(1,-1;Y)},\\
&\underbrace{1}_{\text{TE}(1,0;Y)}=\underbrace{1/2}_{\text{S-CNDE-O}(1,0;Y)}-\underbrace{(-1/2)}_{\text{S-CNIE-O}(0,1;Y)}
=\underbrace{1/2}_{\text{S-CNDE-O}(1,0;Y)}+\underbrace{1/2}_{\text{S-CNIE-O}(1,0;Y)},\\
&\underbrace{1}_{\text{TE}(0,-1;Y)}=\underbrace{1/2}_{\text{S-CNDE-O}(0,-1;Y)}-\underbrace{(-1/2)}_{\text{S-CNIE-O}(-1,0;Y)}
=\underbrace{1/2}_{\text{S-CNDE-O}(0,-1;Y)}+\underbrace{1/2}_{\text{S-CNIE-O}(0,-1;Y)}.
\end{aligned}
\end{equation}
We can verify that the skew-symmetry and additivity hold.
\end{customeg}

\subsection{Identification}
All these direct and indirect effect measures for ordinal treatment are constructed from NDE and NIE, and therefore can be identified from observational distributions through the identification of NDE and NIE, for example via Proposition~\ref{prop1}.
Formally,
\begin{restatable}{theorem}{Theoremten}[Identification of direct and indirect effects for ordinal treatment]
Under SCM $\mathcal{M}_N$ and Assumption \ref{SCAS2}, CNDE-O, CNIE-O, S-NDE, S-NIE, S-CNDE-O and S-CNIE-O are all identifiable {from $\mathbb{P}(X, Y, M)$}.
\end{restatable}

\section{Application}

In this section, we present an application using a real-world dataset.

{\bf Dataset.}
We use the Framingham Heart Study dataset (\url{https://search.r-project.org/CRAN/refmans/riskCommunicator/html/framingham.html}).
The Framingham Heart Study is a landmark prospective study of cardiovascular disease conducted among residents of Framingham, Massachusetts, which was the first to establish the concept of cardiovascular risk factors.
Among the available variables, we consider cigarettes smoked per day as the exposure ($X$), body mass index (BMI) as the mediator ($M$), and heart rate as the outcome ($Y$).
We examine whether the effect of cigarette smoking on heart rate is mediated through BMI.
The analysis is restricted to male participants under 50 years old who smoke and have an education level below high school.
We remove patients with missing values, resulting in a sample size of 224.
We perform 100 bootstrap replications \citep{Efron1979} to obtain the mean and 95\% confidence interval (CI) for each estimator.

{\bf Estimation.}
The direct and indirect effect measures considered in this paper can be estimated by evaluating $\mathbb{E}[Y|X=x,M=m]$ and $\mathfrak{p}(M=m|X=x)$, as indicated in Proposition~\ref{prop1} and Theorem~\ref{theo4}.
We adopt nonlinear models with Gaussian noise, assuming $Y|X=x,M=m \sim \mathcal{N}(\theta(x,m),\sigma_Y)$ and $M|X=x \sim \mathcal{N}(\varphi(x),\sigma_M)$, where $\mathcal{N}(a,b)$ denotes a Gaussian distribution with mean $a$ and standard deviation $b$.
We have $\partial_x\mathbb{E}[Y|X=x,M=m]=\partial_x\theta(x,m)$.
We estimate $\theta$ and $\varphi$ using the local linear estimator \citep{Li2007}, implemented via the R package ``np'' (\url{https://cran.r-project.org/web/packages/np/index.html}).
$\sigma_Y$ and $\sigma_M$ are estimated from the standard deviations of the residuals obtained from the regression models.
Further details of the estimation procedure are provided in Appendix~\ref{SECC}.

\begin{table}[t]
\centering\footnotesize
\caption{Means and 95\% CIs of each estimate. If the quantity is skew-symmetric, we report only one direction.}
\label{tab:my_label1}
\begin{minipage}{0.48\linewidth}
\centering
\begin{tabular}{l|c}
\hline
\multicolumn{2}{l}{\textbf{(1). From 20 to 40}} \\ 
\hline\hline 
$\mathrm{TE}(40,20)$ & $0.524\;[-1.544,\;2.241]$ \\ 
\hline 
$\mathrm{NDE}(40,20)$ & $0.955\;[-1.015,\;2.542]$ \\ 
$\mathrm{NDE}(20,40)$ & $-0.363\;[-2.130,\;1.590]$ \\ 
$\mathrm{NIE}(40,20)$ & $0.161\;[-0.157,\;0.600]$ \\ 
$\mathrm{NIE}(20,40)$ & $0.431\;[-0.191,\;1.179]$ \\ 
\hline 
$\mathrm{CNDE}(40,20)$ & $0.586\;[-1.344,\;2.259]$ \\ 
$\mathrm{CNIE}(40,20)$ & $-0.062\;[-0.486,\;0.350]$ \\ 
\hline 
$\mathrm{S\text{-}NDE}(40,20)$ & $0.659\;[-1.301,\;2.313]$ \\ 
$\mathrm{S\text{-}NIE}(40,20)$ & $-0.135\;[-0.554,\;0.301]$ \\ 
\hline 
$\mathrm{CNDE\text{-}O}(40,20)$ & $0.600\;[-1.330,\;2.269]$ \\ 
$\mathrm{CNDE\text{-}O}(20,40)$ & $-0.572\;[-2.250,\;1.358]$ \\ 
$\mathrm{CNIE\text{-}O}(40,20)$ & $-0.048\;[-0.467,\;0.357]$ \\ 
$\mathrm{CNIE\text{-}O}(20,40)$ & $0.077\;[-0.342,\;0.507]$ \\ 
\hline 
$\mathrm{S\text{-}CNDE\text{-}O}(40,20)$ & $0.586\;[-1.344,\;2.259]$ \\ 
$\mathrm{S\text{-}CNIE\text{-}O}(40,20)$ & $-0.062\;[-0.487,\;0.349]$ \\ 
\hline 
\end{tabular}
\end{minipage}
\begin{minipage}{0.48\linewidth}
\begin{tabular}{l|c}
\hline 
\multicolumn{2}{l}{\textbf{(2). From 20 to 30}} \\ 
\hline\hline 
$\mathrm{TE}(30,20)$ & $0.223\;[-0.691,\;1.112]$ \\ 
\hline 
$\mathrm{NDE}(30,20)$ & $0.228\;[-0.599,\;1.150]$ \\ 
$\mathrm{NDE}(20,30)$ & $-0.113\;[-1.009,\;0.798]$ \\ 
$\mathrm{NIE}(30,20)$ & $0.110\;[-0.066,\;0.369]$ \\ 
$\mathrm{NIE}(20,30)$ & $0.005\;[-0.228,\;0.224]$ \\ 
\hline 
$\mathrm{CNDE}(30,20)$ & $0.163\;[-0.697,\;1.064]$ \\ 
$\mathrm{CNIE}(30,20)$ & $0.060\;[-0.140,\;0.273]$ \\ 
\hline 
$\mathrm{S\text{-}NDE}(30,20)$ & $0.171\;[-0.689,\;1.066]$ \\ 
$\mathrm{S\text{-}NIE}(30,20)$ & $0.052\;[-0.144,\;0.258]$ \\ 
\hline 
$\mathrm{CNDE\text{-}O}(30,20)$ & $0.169\;[-0.687,\;1.073]$ \\ 
$\mathrm{CNDE\text{-}O}(20,30)$ & $-0.158\;[-1.057,\;0.708]$ \\ 
$\mathrm{CNIE\text{-}O}(30,20)$ & $0.066\;[-0.134,\;0.278]$ \\ 
$\mathrm{CNIE\text{-}O}(20,30)$ & $-0.054\;[-0.267,\;0.145]$ \\ 
\hline 
$\mathrm{S\text{-}CNDE\text{-}O}(30,20)$ & $0.163\;[-0.697,\;1.064]$ \\ 
$\mathrm{S\text{-}CNIE\text{-}O}(30,20)$ & $0.060\;[-0.140,\;0.272]$ \\ 
\hline 
\end{tabular}
\end{minipage}\\
\vspace{0.2cm}
\begin{minipage}{0.48\linewidth}
\begin{tabular}{l|c}
\hline \multicolumn{2}{l}{\textbf{(3). From 30 to 40}} \\ 
\hline\hline 
$\mathrm{TE}(40,30)$ & $0.301\;[-0.978,\;1.283]$ \\ 
\hline 
$\mathrm{NDE}(40,30)$ & $0.522\;[-0.624,\;1.462]$ \\ 
$\mathrm{NDE}(30,40)$ & $-0.347\;[-1.315,\;0.877]$ \\ 
$\mathrm{NIE}(40,30)$ & $-0.047\;[-0.252,\;0.159]$ \\ 
$\mathrm{NIE}(30,40)$ & $0.221\;[-0.059,\;0.592]$ \\ 
\hline 
$\mathrm{CNDE}(40,30)$ & $0.423\;[-0.732,\;1.379]$ \\ 
$\mathrm{CNIE}(40,30)$ & $-0.122\;[-0.410,\;0.103]$ \\ 
\hline 
$\mathrm{S\text{-}NDE}(40,30)$ & $0.435\;[-0.732,\;1.389]$ \\ 
$\mathrm{S\text{-}NIE}(40,30)$ & $-0.134\;[-0.421,\;0.095]$ \\ 
\hline 
$\mathrm{CNDE\text{-}O}(40,30)$ & $0.431\;[-0.720,\;1.387]$ \\ 
$\mathrm{CNDE\text{-}O}(30,40)$ & $-0.414\;[-1.372,\;0.750]$ \\ 
$\mathrm{CNIE\text{-}O}(40,30)$ & $-0.113\;[-0.393,\;0.110]$ \\ 
$\mathrm{CNIE\text{-}O}(30,40)$ & $0.131\;[-0.097,\;0.427]$ \\ 
\hline 
$\mathrm{S\text{-}CNDE\text{-}O}(40,30)$ & $0.423\;[-0.733,\;1.379]$ \\ 
$\mathrm{S\text{-}CNIE\text{-}O}(40,30)$ & $-0.122\;[-0.410,\;0.103]$ \\ 
\hline
\end{tabular}
\end{minipage}
\end{table}

{\bf Results.}
Table~\ref{tab:my_label1} reports the estimated means and bootstrap 95\% CIs for the direct and indirect effect measures under the following comparisons: (1) $20$ versus $40$, (2) $20$ versus $30$, and (3) $30$ versus $40$ cigarettes per day. 
Since NDE, NIE, CNDE-O, and CNIE-O are not skew-symmetric, both the $(x'',x')$ and $(x',x'')$ forms are reported for these measures. 
Note that the actual domain of $X$ consists of the discrete integer values from $0$ to $90$. Nevertheless, we fit models that treat $X$ as a continuous variable. CNDE and CNIE are computed using derivatives of the fitted models with respect to $X$. NDE, NIE, S-NDE, and S-NIE are computed by evaluating the fitted models at the corresponding $X$ values. 
CNDE-O, CNIE-O, S-CNDE-O, and S-CNIE-O are computed by treating $X$ as an ordinal variable taking values in $\{0, 1, \ldots, 90\}$. 
{We observe that the estimated S-CNDE-O and S-CNIE-O values are nearly identical to those of CNDE and CNIE, respectively. This is not surprising, since the support of $X$ forms a fine grid, making the ordinal cumulative effects close approximations to their continuous counterparts.} 

The TE estimates are positive for all three contrasts:
$\mathrm{TE}(40,20) = 0.524$, $\mathrm{TE}(30,20) = 0.223$, and
$\mathrm{TE}(40,30) = 0.301$. $\mathrm{TE}(40,20) = \mathrm{TE}(30,20) + \mathrm{TE}(40,30)$, 
confirming the additivity of TE.

{
The NDE and NIE estimates can exhibit paradoxical mediation patterns due to their failure to satisfy skew-symmetry and additivity. For example, $\mathrm{NDE}(40,20)=0.955$ is substantially larger in magnitude than  $\mathrm{NDE}(20,40)=-0.363$, while $\mathrm{NIE}(40,20)=0.161$ is considerably smaller than $\mathrm{NIE}(20,40)=0.431$. Thus, simply reversing the treatment contrast leads to markedly different effect sizes, making interpretation difficult. A similar phenomenon arises when comparing the decomposition of the total effect under opposite treatment contrasts. We have 
\begin{align}
&\underbrace{0.223}_{\text{TE}(30,20)}=\underbrace{0.228}_{\text{NDE}(30,20)}-\underbrace{0.005}_{\text{NIE}(20,30)},
\end{align}
where the total effect is almost entirely attributable to the direct effect. In contrast,
\begin{align}
&\underbrace{-0.223}_{\text{TE}(20,30)}=\underbrace{-0.113}_{\text{NDE}(20,30)}-\underbrace{0.110}_{\text{NIE}(30,20)},
\end{align}
for which the direct and indirect effects contribute approximately equally. These results appear paradoxical because simply reversing the direction of the treatment contrast leads to substantially different conclusions regarding the relative importance of the direct and indirect pathways.} 

{
The lack of additivity is also problematic. For example, 
$\mathrm{NIE}(20,30) + \mathrm{NIE}(30,40) = 0.005 + 0.221 = 0.226$, 
which is substantially smaller than $\mathrm{NIE}(20,40)=0.431$. This result is difficult to interpret, as the estimated indirect effect depends strongly on whether an intermediate treatment level is considered. In other words, the indirect effect for the contrast from 40 to 20 cigarettes per day is not consistent with the sum of the corresponding indirect effects over the adjacent contrasts 40--30 and 30--20, violating a natural notion of additivity.
}

The proposed alternative measures clarify these patterns.
The skew-symmetrized measures S-NDE and S-NIE remove the directional asymmetry of NDE and NIE, while the cumulative measures restore additivity over adjacent treatment intervals.
In particular, additivity holds for CNDE, CNIE, CNDE-O, CNIE-O, S-CNDE-O, and S-CNIE-O.
The CNDE and CNIE estimates indicate a moderate positive direct component with a small negative indirect component for the 20--40 contrast, a positive direct component with a small positive indirect component for the 20--30 contrast, and a positive direct component with a small negative indirect component for the 30--40 contrast.
These decompositions are coherent across intervals: the direct and indirect components of the effect of increasing cigarette consumption from 20 to 40 cigarettes per day on heart rate are equal to  the corresponding sums of the 20--30 and 30--40 components.


\section{Conclusions}

We show that the standard decomposition of the TE into NDE and NIE can give rise to paradoxical interpretations. These paradoxes stem from the failure of NDE and NIE to satisfy two fundamental properties of interpretable causal effects—skew-symmetry and additivity—and can complicate the interpretation of mediation analysis, particularly when comparing effects across treatment contrasts or aggregating effects over multiple levels. Despite their importance, these issues have received limited attention in the literature, possibly because earlier approaches, such as the regression-based framework of \citet{Baron1986}, do not exhibit such behavior under linear models.

To address these limitations, we proposed a new decomposition of the TE based on cumulative natural direct and indirect effects that preserves skew-symmetry and additivity. This formulation provides a coherent and intuitive interpretation of direct and indirect effects across treatment levels and avoids the paradoxes associated with NDE and NIE. 
Our results contribute to the conceptual foundations of mediation analysis by emphasizing the importance of structural properties in effect decomposition. The proposed framework offers a principled alternative for settings with continuous or ordinal treatments, where path-based interpretations are natural and desirable.

The identification results in the paper rely on the sequential ignorability assumption \citep{Imai2010a,Imai2010b}. 
In settings where this assumption is not justified \citep{Ding2016,Miles2019}, sensitivity analysis plays a crucial role in assessing the robustness of conclusions to unmeasured confounding \citep{Imai2010b,Ding2016}. 
Furthermore, instrumental-variable, proximal-variable, and difference-in-differences approaches can facilitate identification in the presence of unmeasured confounding \citep{Rudolph2021,Dukes2023,Huber2026}.
Extending these approaches to our proposed measures is an important avenue for future research. 
Another important extension concerns settings with multiple mediators, which are common in applied studies \citep{Daniel2015,Xia2022,Zhou2023}. Developing cumulative direct and indirect effect measures in the presence of multiple mediators, as well as studying their identification and interpretation, is an interesting direction for future work. 
Efficient estimation is another important topic, particularly in settings with limited sample sizes.
\citet{Tchetgen2012} developed semiparametric theory for causal mediation analysis, while \citet{Huang2024} proposed generalized weighting estimators. Developing efficient estimators for our proposed measures is an important direction for future research.

We hope that the proposed framework will provide a useful and interpretable alternative for mediation analysis and stimulate further research on structurally coherent decompositions of causal effects.

\begin{appendices}



\section{Appendix: Proofs}
\label{SECA}

In the appendix, we provide proofs of the lemmas and theorems in the body of the paper.

\Lemmaone*


\begin{proof}
Under SCM $\mathcal{M}_{N}$ and Assumption \ref{SCAS1}, we consider the function $\mathbb{E}[Y_{x,M_{x^{*}}}]$.
From Assumption \ref{SCAS1}, $\mathbb{E}[Y_{x,M_{x^{*}}}]$ is totally differentiable \citep{Abraham2012} and we have
\begin{equation}
\begin{aligned}
&\mathbb{E}[Y_{x+h,M_{x^{*}}+h}]-\mathbb{E}[Y_{x,M_{x^{*}}}]=h\partial_{x}\mathbb{E}[Y_{x,M_{x^{*}}}]+h\partial_{x^{*}}\mathbb{E}[Y_{x,M_{x^{*}}}]+o(h)
\end{aligned}
\end{equation}
for any $x, x^{*} \in \Omega_X$.
Thus, we have 
\begin{equation}
\begin{aligned}
&\underbrace{\partial_x \mathbb{E}[Y_x]}_{\text{LTE}(x;Y)}= \underbrace{\partial_{x^*} \mathbb{E}[Y_{x,{M}_{x^*}}]|_{x^*=x}}_{\text{LNIE}(x;Y)}+\underbrace{\partial_{x^*} \mathbb{E}[Y_{x^*,{M}_x}]|_{x^*=x}}_{\text{LNDE}(x;Y)}.
\end{aligned}
\end{equation}
for any $x \in \Omega_X$.
\end{proof}

\Theoremone*


\begin{proof}
Let $\mathbbm{1}$ be a delta function.
Under SCM $\mathcal{M}_{N}$, we have
\begin{equation}
\label{A5}
\begin{aligned}
&Y_{x,{M}_{x+\delta}}-Y_{x,{M}_{x}}=\int_{\Omega_{M}}Y_{x,m}\{\mathbbm{1}({M}_{x+\delta}=m)-\mathbbm{1}({M}_{x}=m)\}dm.
\end{aligned}
\end{equation}
for almost every subject. 
Under Assumption \ref{SCAS1}, 
taking expectations on both sides of Eq.~\eqref{A5}, we have
\begin{align}
\label{eq75}
\mathbb{E}[Y_{x,{M}_{x+\delta}}]-\mathbb{E}[Y_{x,{M}_{x}}]
&=\int_{\Omega_{M}}\mathbb{E}[Y_{x,m}]\{\mathfrak{p}({M}_{x+\delta}=m)-\mathfrak{p}({M}_{x}=m)\}dm
\end{align}
Assuming the regularity conditions hold, integration by parts of Eq.~\eqref{eq75} yields 
\begin{align}
\mathbb{E}[Y_{x,{M}_{x+\delta}}]-\mathbb{E}[Y_{x,{M}_{x}}]
&=-\int_{\Omega_{M}}\partial_{m}\mathbb{E}[Y_{x,m}]\{\mathbb{P}({M}_{x+\delta}\leq m)-\mathbb{P}({M}_{x}\leq m)\}dm\\
&=\int_{\Omega_{M}}\partial_{m}\mathbb{E}[Y_{x,m}]\{\mathbb{P}({M}_{x+\delta}> m)-\mathbb{P}({M}_{x}>m)\}dm
\end{align}
for any $x \in \Omega_X$.
Then, taking limit $\delta \rightarrow 0$, we have       
\begin{equation}
\begin{aligned}
&\partial_{x^*}\mathbb{E}[Y_{x,M_{x^*}}]|_{x^*=x}=\int_{\Omega_{M}}\partial_{m}\mathbb{E}[Y_{x,m}]\partial_x\mathbb{P}({M}_{x}> m)dm
\end{aligned}
\end{equation}   
for any $x \in \Omega_X$.
Under Assumption 1, we have
\begin{equation}
\begin{aligned}
&\partial_{x^*} \mathbb{E}[Y_{x,{M}_{x^*}}]|_{x^*=x}=\int_{\Omega_M}\partial_{m}\mathbb{E}[Y|X=x,M=m]\partial_x \mathbb{P}(M>m|x)dm.
\end{aligned}
\end{equation}
for any $x \in \Omega_X$.
Also, under Assumption \ref{SCAS2}, we have $\partial_{x^*} \mathbb{E}[Y_{x^*,{M}_{x}}]=\partial_{x^*} \mathbb{E}[Y_{x^*,m}|{M}_{x}]=\partial_{x^*} \int_{\Omega_M}\mathbb{E}[Y_{x^*,m}|M=m]\mathfrak{p}({M}=m|X=x)dm$, then we have
\begin{equation}
\begin{aligned}
&\partial_{x^*} \mathbb{E}[Y_{x^*,{M}_x}]|_{x^*=x}=\int_{\Omega_{M}} \partial_x \mathbb{E}[Y|X=x,M=m]\mathfrak{p}(M=m|X=x)d{m}.
\end{aligned}
\end{equation}
for any $x \in \Omega_X$.
\end{proof}

\Theoremtwo*


\begin{proof}
We have
\begin{align}
&\text{CNDE}(x'',x';Y)=\int_{x'}^{x''} \text{LNDE}(x;Y) dx=-\int_{x''}^{x'} \text{LNDE}(x;Y) dx=-\text{CNDE}(x',x'';Y),\\
&\text{CNIE}(x'',x';Y)=\int_{x'}^{x''}\text{LNIE}(x;Y)dx=-\int_{x''}^{x'}\text{LNIE}(x;Y)dx=-\text{CNIE}(x',x'';Y),\\
&\text{CNDE}(x''',x'';Y)+\text{CNDE}(x'',x';Y)=\int_{x''}^{x'''} \text{LNDE}(x;Y) dx+\int_{x'}^{x''} \text{LNDE}(x;Y) dx\\
&=\int_{x'}^{x'''} \text{LNDE}(x;Y) dx=\text{CNDE}(x''',x';Y),\\
&\text{CNIE}(x''',x'';Y)+\text{CNIE}(x'',x';Y)=\int_{x''}^{x'''} \text{LNIE}(x;Y) dx+\int_{x'}^{x''} \text{LNIE}(x;Y) dx\\
&=\int_{x'}^{x'''} \text{LNIE}(x;Y) dx=\text{CNIE}(x''',x';Y).
\end{align}
Then, CNDE and CNIE have skew-symmetry and additivity.
\end{proof}

\Theoremthree*


\begin{proof}
This is because 
\begin{align}
\text{TE}(x'',x';Y)
&=\int_{x'}^{x''}\text{LTE}(x;Y)dx\\
&=\int_{x'}^{x''}\{\text{LNDE}(x;Y)+\text{LNIE}(x;Y)\}dx\\
&=\int_{x'}^{x''}\text{LNDE}(x;Y)dx+\int_{x'}^{x''}\text{LNIE}(x;Y)dx\\
&=\text{CNDE}(x'',x';Y)+\text{CNIE}(x'',x';Y)\\
&=\text{CNDE}(x'',x';Y)-\text{CNIE}(x',x'';Y)
\end{align}
holds for any $x', x'' \in \Omega_X$.
From the symmetry, we have $\text{CNIE}(x'',x';Y)=-\text{CNIE}(x',x'';Y)$, then the last equation holds.
\end{proof}

\Theoremfour*


\begin{proof}
This is because LNDE and LNIE are identifiable, for any $x', x'' \in \Omega_X$, as:
\begin{equation}
\begin{aligned}
&\text{LNDE}(x;Y)=\int_{\Omega_{M}} \partial_x \mathbb{E}[Y|X=x,M=m]\mathfrak{p}(M=m|X=x)d{m},
\end{aligned}
\end{equation}
\begin{equation}
\begin{aligned}
&\text{LNIE}(x;Y)=\int_{\Omega_M}\partial_{m}\mathbb{E}[Y|X=x,M=m]\partial_x \mathbb{P}(M>m|X=x)dm.
\end{aligned}
\end{equation}
By substituting it into 
\begin{equation}
\begin{aligned}
&\text{CNDE}(x'',x';Y)\defeq\int_{x'}^{x''} \text{LNDE}(x;Y) dx,\\
&\text{CNIE}(x'',x';Y)\defeq\int_{x'}^{x''}\text{LNIE}(x;Y)dx,
\end{aligned}
\end{equation}
we have the results.
\end{proof}

\Theoremfive*

\begin{proof}
Let $x'=x_i$, $x''=x_j$, and $x'''=x_\ell$.
First suppose $i<j<\ell$. Then, by the definition of CNDE-O,
\begin{align}
\text{CNDE-O}(x''',x';Y)
&=\text{CNDE-O}(x_\ell,x_i;Y)\\
&=\sum_{k=i}^{\ell-1}\text{NDE}(x_{k+1},x_k;Y)\\
&=\sum_{k=i}^{j-1}\text{NDE}(x_{k+1},x_k;Y)+\sum_{k=j}^{\ell-1}\text{NDE}(x_{k+1},x_k;Y)\\
&=\text{CNDE-O}(x_j,x_i;Y)+\text{CNDE-O}(x_\ell,x_j;Y)\\
&=\text{CNDE-O}(x'',x';Y)
+\text{CNDE-O}(x''',x'';Y)
\end{align}
Similarly,
\begin{align}
\text{CNIE-O}(x''',x';Y)
&=\sum_{k=i}^{\ell-1}\text{NIE}(x_{k+1},x_k;Y)\\
&=\sum_{k=i}^{j-1}\text{NIE}(x_{k+1},x_k;Y)+\sum_{k=j}^{\ell-1}\text{NIE}(x_{k+1},x_k;Y)\\
&=\text{CNIE-O}(x'',x';Y)+\text{CNIE-O}(x''',x'';Y).
\end{align}
Next suppose $i>j>\ell$. Then, by the definition of CNDE-O,
\begin{align}
\text{CNDE-O}(x''',x';Y)
&=\text{CNDE-O}(x_\ell,x_i;Y)\\
&=\sum_{k=\ell}^{i-1}\text{NDE}(x_k,x_{k+1};Y)\\
&=\sum_{k=j}^{i-1}\text{NDE}(x_k,x_{k+1};Y)+\sum_{k=\ell}^{j-1}\text{NDE}(x_k,x_{k+1};Y)\\
&=\text{CNDE-O}(x_j,x_i;Y)+\text{CNDE-O}(x_\ell,x_j;Y)\\
&=\text{CNDE-O}(x'',x';Y)+\text{CNDE-O}(x''',x'';Y).
\end{align}
Similarly,
\begin{align}
\text{CNIE-O}(x''',x';Y)
&=\sum_{k=\ell}^{i-1}\text{NIE}(x_k,x_{k+1};Y)\\
&=\sum_{k=j}^{i-1}\text{NIE}(x_k,x_{k+1};Y)+\sum_{k=\ell}^{j-1}\text{NIE}(x_k,x_{k+1};Y)\\
&=\text{CNIE-O}(x'',x';Y)+\text{CNIE-O}(x''',x'';Y).
\end{align}
\end{proof}

\Theoremsix*

\begin{proof}
Let $x'=x_i$ and $x''=x_j$.
If $i<j$, we have 
\begin{align}
&\text{CNDE-O}(x'',x';Y)-\text{CNIE-O}(x',x'';Y)\\
&=\text{CNDE-O}(x_j,x_i;Y)-\text{CNIE-O}(x_i,x_j;Y)\\
&=\sum_{k=i}^{j-1} \text{NDE}(x_{k+1},x_k;Y)-\sum_{k=i}^{j-1} \text{NIE}(x_k,x_{k+1};Y)\\
&=\sum_{k=i}^{j-1}\{\text{NDE}(x_{k+1},x_k;Y)-\text{NIE}(x_k,x_{k+1};Y)\}\\
&=\sum_{k=i}^{j-1}\text{TE}(x_{k+1},x_k;Y)=\text{TE}(x'',x';Y).
\end{align}
If $i>j$, we have 
\begin{equation}
\begin{aligned}
&\text{CNDE-O}(x'',x';Y)-\text{CNIE-O}(x',x'';Y)\\
&=\text{CNDE-O}(x_j,x_i;Y)-\text{CNIE-O}(x_i,x_j;Y)\\
&=\sum_{k=j}^{i-1} \text{NDE}(x_k,x_{k+1};Y)-\sum_{k=j}^{i-1} \text{NIE}(x_{k+1},x_k;Y)\\
&=\sum_{k=j}^{i-1}\{\text{NDE}(x_k,x_{k+1};Y)-\text{NIE}(x_{k+1},x_k;Y)\}\\
&=\sum_{k=j}^{i-1}\text{TE}(x_k,x_{k+1};Y)=\text{TE}(x'',x';Y).
\end{aligned}
\end{equation}
\end{proof}

\Theoremseven*

\begin{proof}
We have 
\begin{align}
&\text{S-NDE}(x'',x';Y)-\text{S-NIE}(x',x'';Y)\\
&= \frac{1}{2}\Big(\text{NDE}(x'',x';Y) - \text{NDE}(x',x'';Y)\Big)-\frac{1}{2}\Big(\text{NIE}(x',x'';Y) - \text{NIE}(x'',x';Y)\Big)\\
&=\frac{1}{2}\Big(\text{NDE}(x'',x';Y) - \text{NIE}(x',x'';Y)\Big)-\frac{1}{2}\Big(\text{NDE}(x',x'';Y) - \text{NIE}(x'',x';Y)\Big)\\
&=\frac{1}{2}\text{TE}(x'',x';Y)-\frac{1}{2}\text{TE}(x',x'';Y)\\
&=\frac{1}{2}\text{TE}(x'',x';Y)+\frac{1}{2}\text{TE}(x'',x';Y)=\text{TE}(x'',x';Y).
\end{align}
\end{proof}

\Theoremeight*

\begin{proof}
Let $x'=x_i$, $x''=x_j$, and $x'''=x_\ell$.
First, we show skew-symmetry. If $i<j$, then 
\begin{align}
&\text{S-CNDE-O}(x_j,x_i;Y)=\sum_{k=i}^{j-1}\text{S-NDE}(x_{k+1},x_k;Y)\\
&=-\sum_{k=i}^{j-1}\text{S-NDE}(x_k,x_{k+1};Y)=-\text{S-CNDE-O}(x_i,x_j;Y),
\end{align}
where the second equality follows from the skew-symmetry of S-NDE. Similarly, 
\begin{align}
&\text{S-CNIE-O}(x_j,x_i;Y)=\sum_{k=i}^{j-1}\text{S-NIE}(x_{k+1},x_k;Y)\\
&=-\sum_{k=i}^{j-1}\text{S-NIE}(x_k,x_{k+1};Y)=-\text{S-CNIE-O}(x_i,x_j;Y).
\end{align}
The case $i>j$ follows analogously.
Next, we show additivity. Suppose $i<j<\ell$. Then 
\begin{align}
&\text{S-CNDE-O}(x_\ell,x_i;Y)=\sum_{k=i}^{\ell-1}\text{S-NDE}(x_{k+1},x_k;Y)\\
&=\sum_{k=i}^{j-1}\text{S-NDE}(x_{k+1},x_k;Y)+\sum_{k=j}^{\ell-1}\text{S-NDE}(x_{k+1},x_k;Y)\\
&=\text{S-CNDE-O}(x_j,x_i;Y)+\text{S-CNDE-O}(x_\ell,x_j;Y).
\end{align}
Similarly, 
\begin{align}
&\text{S-CNIE-O}(x_\ell,x_i;Y)=\sum_{k=i}^{\ell-1}\text{S-NIE}(x_{k+1},x_k;Y)\\
&=\sum_{k=i}^{j-1}\text{S-NIE}(x_{k+1},x_k;Y)+\sum_{k=j}^{\ell-1}\text{S-NIE}(x_{k+1},x_k;Y)\\
&=\text{S-CNIE-O}(x_j,x_i;Y)+\text{S-CNIE-O}(x_\ell,x_j;Y).
\end{align}
Finally, suppose $i>j>\ell$. Then 
\begin{align}
&\text{S-CNDE-O}(x_\ell,x_i;Y)=\sum_{k=\ell}^{i-1}\text{S-NDE}(x_k,x_{k+1};Y)\\
&=\sum_{k=j}^{i-1}\text{S-NDE}(x_k,x_{k+1};Y)+\sum_{k=\ell}^{j-1}\text{S-NDE}(x_k,x_{k+1};Y)\\
&=\text{S-CNDE-O}(x_j,x_i;Y)+\text{S-CNDE-O}(x_\ell,x_j;Y).
\end{align}
Similarly, 
\begin{align}
&\text{S-CNIE-O}(x_\ell,x_i;Y)=\sum_{k=\ell}^{i-1}\text{S-NIE}(x_k,x_{k+1};Y)\\
&=\sum_{k=j}^{i-1}\text{S-NIE}(x_k,x_{k+1};Y)+\sum_{k=\ell}^{j-1}\text{S-NIE}(x_k,x_{k+1};Y)\\
&=\text{S-CNIE-O}(x_j,x_i;Y)+\text{S-CNIE-O}(x_\ell,x_j;Y).
\end{align}
Therefore, S-CNDE-O and S-CNIE-O satisfy skew-symmetry and additivity.
\end{proof}

\Theoremnine*

\begin{proof}
Let $x'=x_i$ and $x''=x_j$. If $i<j$, then
\begin{align}
&\text{S-CNDE-O}(x'',x';Y)-\text{S-CNIE-O}(x',x'';Y)\\
&=\text{S-CNDE-O}(x_j,x_i;Y)-\text{S-CNIE-O}(x_i,x_j;Y)\\
&=\sum_{k=i}^{j-1}\text{S-NDE}(x_{k+1},x_k;Y)-\sum_{k=i}^{j-1}\text{S-NIE}(x_k,x_{k+1};Y)\\
&=\sum_{k=i}^{j-1}\{\text{S-NDE}(x_{k+1},x_k;Y)+\text{S-NIE}(x_{k+1},x_k;Y)\}\\
&=\sum_{k=i}^{j-1}\text{TE}(x_{k+1},x_k;Y)
=\text{TE}(x_j,x_i;Y)=\text{TE}(x'',x';Y).
\end{align}
If $i>j$, then
\begin{align}
&\text{S-CNDE-O}(x'',x';Y)-\text{S-CNIE-O}(x',x'';Y)\\
&=\text{S-CNDE-O}(x_j,x_i;Y)-\text{S-CNIE-O}(x_i,x_j;Y)\\
&=\sum_{k=j}^{i-1}\text{S-NDE}(x_k,x_{k+1};Y)-\sum_{k=j}^{i-1}\text{S-NIE}(x_{k+1},x_k;Y)\\
&=\sum_{k=j}^{i-1}\{\text{S-NDE}(x_k,x_{k+1};Y)+\text{S-NIE}(x_k,x_{k+1};Y)\}\\
&=\sum_{k=j}^{i-1}\text{TE}(x_k,x_{k+1};Y)=\text{TE}(x_j,x_i;Y)=\text{TE}(x'',x';Y).
\end{align}
\end{proof}

\Theoremten*

\begin{proof}
This theorem is directly from Proposition \ref{prop1}.
\end{proof}

\section{Appendix: Covariate Adjustment}
\label{SECB}

\begin{figure}[tb]
\centering
\scalebox{1}{
\begin{tikzpicture}
\node (z) at (0,2) {$M$};

\node (w) at (0,-2) {$W$};
\node (y) at (2,0) {$Y$};
\node (x) at (-2,0) {$X$};

\path (x) edge[->] (y);
    
\path (z) edge[->] (y);
\path (x) edge[->] (z);

\path (w) edge[->] (y);
\path (w) edge[->] (x);

\path (w) edge[->] (z);

\path (x) edge[<->,dotted,bend left] (z);
\path (z) edge[<->,dotted,bend left] (y);
\path (x) edge[<->,dotted,bend right] (y);

\path (w) edge[<->,dotted,bend left] (x);
\path (w) edge[<->,dotted,bend right] (y);
\path (w) edge[<->,dotted,bend left] (z);


\end{tikzpicture}
}
\vspace{-0cm}
\caption{The causal graph representing SCM $\mathcal{M}_{C}$.}
\label{DAGN2}
\end{figure}
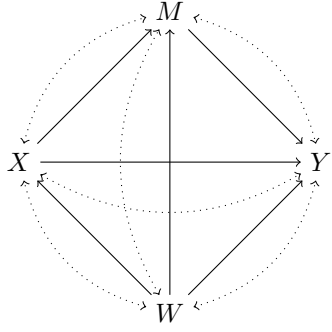

We consider the following nonparametric SCM with observed confounder $W$, $\mathcal{M}_{C}$: 
\begin{equation}
\begin{gathered}
  \mathcal{M}_{C}: W:=f_W(\epsilon_W),\ \ 
  X:=f_X(W,\epsilon_X),\ \ 
    M:=f_{M}(X,W,\epsilon_{M}),\ \
    Y:=f_Y(X,M,W,\epsilon_Y),
   \end{gathered}
\end{equation}
where $\epsilon_W$, $\epsilon_X$, $\epsilon_{M}$, and $\epsilon_Y$ are latent exogenous variables. 
Figure \ref{DAGN2} shows the causal graph representing $\mathcal{M}_{C}$, where bidirected edges indicate unmeasured exogenous confounders that affect the variables. 

We assume the conditional version of Assumptions \ref{SCAS2} and \ref{SCAS1}.
\begin{assumption}[Conditional Sequential Ignorability]
\label{Con_SCAS2}
The following two statements of conditional independence hold:
$\{Y_{x,m},M_{x'}\} \indep X|W\text{ and }M_{x'} \indep Y_{x,m}|W$
for any $m \in \Omega_M$ and $x, x' \in \Omega_X$, where $\mathfrak{p}(X=x|W)>0$ and $\mathfrak{p}(M=m|X=x,W)>0$ for any $m \in \Omega_M$ and $x \in \Omega_X$.
\end{assumption}
\begin{assumption}[Conditional Differentiability and Continuity]
\label{Con_SCAS1}
For any $x,x^* \in \Omega_X$, $\mathbb{E}[Y_{x,{M}_{x^*}}|W]$ is partially differentiable w.r.t. $x$ and $x^*$ respectively, and 
$\partial_{x^*} \mathbb{E}[Y_{x,{M}_{x^*}}|W]$ and $\partial_{x^*} \mathbb{E}[Y_{x^*,{M}_x}|W]$ 
are continuous w.r.t. $x$.
\end{assumption}
 
\citet{Imai2010a} present the following proposition concerning covariates:
\begin{proposition}
\label{prop2}
Under SCM $\mathcal{M}_C$ and Assumption \ref{Con_SCAS2}, the conditional expectation of the counterfactual $\mathbb{E}[Y_{x,M_{x'}}|W=w]$ is identifiable by
\begin{equation}
\begin{aligned}
&\mathbb{E}[Y_{x,M_{x'}}|W=w]=\int_{\Omega_{M}}\mathbb{E}[Y|X=x,M=m,W=w]\mathfrak{p}(M=m|X=x',W=w)d{m}
\end{aligned}
\end{equation}
for any $x, x' \in \Omega_X$ and $w \in \Omega_W$.
\end{proposition}

We obtain the covariate-adjusted versions of Theorems \ref{cor3} and \ref{theo4}.
\begin{theorem}[Identification of LNDE and LNIE]
\label{Con_cor3}
Under SCM $\mathcal{M}_C$ and Assumptions \ref{Con_SCAS2} and \ref{Con_SCAS1}, 
if the following regularity condition holds: either
 (i) $\mathbb{E}[Y|X=x,M=m,W=w]$ is bounded and $(M,W)$ has bounded support; or
(ii) $\mathbb{E}[Y|X=x, M=m,W=w]$ grows at most polynomially in $(m,w)$, and the exogenous error $\varepsilon_M$ in $\mathcal{M}_N$ 
has a continuously differentiable density $\mathfrak{p}$ such that $\mathfrak{p}'(m)$ is integrable and $\mathfrak{p}(m) \to 0$
 as $|m| \to \infty$ faster than any polynomial (e.g., Gaussian),
 we can identify LNDE and LNIE, for any $x \in \Omega_X$, as:
\begin{equation}
\begin{aligned}
&\text{LNDE}(x;Y)=\int_{\Omega_{W}} \int_{\Omega_{M}} \partial_x \mathbb{E}[Y|X=x,M=m,W=w]\\
&\times\mathfrak{p}(M=m|X=x,W=w)\mathfrak{p}(W=w)d{m}dw,
\end{aligned}
\end{equation}
\begin{equation}
\begin{aligned}
&\text{LNIE}(x;Y)=\int_{\Omega_{W}} \int_{\Omega_M}\partial_{m}\mathbb{E}[Y|X=x,M=m,W=w]\\
&\times\partial_x \mathbb{P}(M>m|X=x,W=w)\mathfrak{p}(W=w)d{m}dw
\end{aligned}
\end{equation}
\end{theorem}
\begin{theorem}[Identification of CNDE and CNIE]
\label{Con_theo4}
Under SCM $\mathcal{M}_{C}$ and Assumptions \ref{Con_SCAS2} and \ref{Con_SCAS1}, CNDE and CNIE are identifiable, for any $x', x'' \in \Omega_X$, as:
\begin{equation}
\begin{aligned}
&\text{CNDE}(x'',x';Y)= \int_{x'}^{x''}\int_{\Omega_{W}}\int_{\Omega_{M}}\partial_x \mathbb{E}[Y|X=x,M=m,W=w]\\
&\times\mathfrak{p}(M=m|X=x,W=w)\mathfrak{p}(W=w)d{m}dwdx,\\
&\text{CNIE}(x'',x';Y)=\int_{x'}^{x''}\int_{\Omega_{W}} \int_{\Omega_M}\partial_{m}\mathbb{E}[Y|X=x,M=m,W=w]\\
&\times\partial_x \mathbb{P}(M>m|X=x,W=w)\mathfrak{p}(W=w)d{m}dwdx.
\end{aligned}
\end{equation}
\end{theorem}

\section{Appendix: Supplemental Information for the Application}
\label{SECC}
We provide supplemental information for the application presented in the main body of the paper.

{\bf Details of Estimation Method.}
We denote the samples as \begin{equation}
\mathcal{D}=\{(Y_1,X_1,M_1,W_1),\dots,(Y_N,X_N,M_N,W_N)\},    
\end{equation}
whose sample size is $N$.
In the main analysis, we set $W = \emptyset$.
Therefore, we assume a nonlinear model with Gaussian noise, i.e., $Y|X=x,M=m,W=w \sim \mathcal{N}(\theta(x,m,w),\sigma_Y^2)$ and $M|X=x,W=w \sim \mathcal{N}(\varphi(x,w),\sigma_M^2)$.

{\bf Step 1.}
We regress $Y$ on $X$, $M$, and $W$ to estimate $\theta$, and regress $M$ on $X$ and $W$ to estimate $\varphi$, using local linear regressions (with interaction term $XM$) with bandwidth $h$ \citep{Li2007}.
Local linear regression is a nonparametric smoothing method that fits linear models locally around each target point using kernel weights to flexibly estimate conditional relationships without assuming a specific functional form.
We estimate $\sigma_M$ from the residuals of the regression of $M$ on $X$ and $W$.
To obtain stable estimates, we choose a bandwidth four times larger than the cross-validated optimum Silverman rule-of-thumb bandwidth \citep{Silverman2018}.
Because our objective is to approximate the distribution of potential outcomes, we prioritize estimation stability over predictive accuracy for the observed outcomes.

{\bf Step 2.}
We generate $\{\tilde{x}_1,\dots,\tilde{x}_K\}$ from a uniform distribution $\text{Unif}(x',x'')$ to approximate $\int_{x'}^{x''}\cdot\, dx$ using Monte Carlo integration.
In our simulation, we set $K = 30$.
We generate $\{\tilde{m}^{k,i}_1,\dots,\tilde{m}^{k,i}_D\}$ from $\mathcal{N}(\hat{\varphi}_h(\tilde{x}_k,W_i),\hat{\sigma}_M^2)$, where $\hat{\varphi}_h(x,w)$ and $\hat{\sigma}_M^2$ are the estimates obtained in Step 1.
In our simulation, we set $D$ as the sample size.
The estimate of CNDE is given as
\begin{equation}
\sum_{i=1}^N\sum_{k=1}^{K}\sum_{d_k=1}^{D} \frac{1}{K}\frac{1}{D}\frac{1}{N}(x''-x')\partial_x \hat{\theta}_h(\tilde{x}_k,\tilde{m}^{k,i}_{d_k},W_i),
\end{equation}
where $\hat{\theta}_h$ is the estimate obtained in Step 1.
We compute $\partial_x\hat{\theta}_h$ using a numerical derivative with a step size of 5.
The estimator is consistent as $N, D, K \to \infty$ and the bandwidth $h \to 0$ with $Nh \to \infty$, under regularity conditions ensuring uniform convergence of the local linear estimators and convergence of the Monte Carlo approximation to the corresponding integrals.
The computational complexity is $\mathcal{O}(NKD + N^2)$, where the first term corresponds to the Monte Carlo integration in Step 2 and the second term arises from the local linear regressions in Step 1.
Other quantities are obtained in a similar way.


\end{appendices}






\bibliographystyle{plainnat}
\bibliography{example}

@book{Pearl09,
 author = {Judea Pearl},
 title = {Causality: Models, Reasoning and Inference},
 publisher = {Cambridge University Press}, 
 year = {2009},
 edition = {2nd}
}

@INCOLLECTION{Chamberlain1984,
title = {Panel data},
author = {Chamberlain, Gary},
year = {1984},
chapter = {22},
pages = {1247-1318},
booktitle = {Handbook of Econometrics},
editor = {Griliches†, Z. and Intriligator, M. D.},
volume = {2},
edition = {1},
publisher = {Elsevier},
url = {https://EconPapers.repec.org/RePEc:eee:ecochp:2-22}
}

@inProceedings{Kawakami2023,
  title = 	 {Instrumental Variable Estimation of Average Partial Causal Effects},
  author =       {Kawakami, Yuta and Kuroki, Manabu and Tian, Jin},
  booktitle = 	 {Proceedings of the 40th International Conference on Machine Learning},
  pages = 	 {16097--16130},
  year = 	 {2023},
  volume = 	 {202},
  series = 	 {Proceedings of Machine Learning Research},
  month = 	 {23--29 Jul},
  publisher =    {PMLR},
}

@inproceedings{Pearl2001,
author = {Pearl, Judea},
title = {Direct and Indirect Effects},
year = {2001},
isbn = {1558608001},
publisher = {Morgan Kaufmann Publishers Inc.},
address = {San Francisco, CA, USA},
booktitle = {Proceedings of the Seventeenth Conference on Uncertainty in Artificial Intelligence},
pages = {411–420},
numpages = {10},
location = {Seattle, Washington},
series = {UAI'01}
}

@article{Imai2010a,
	author = {Imai, Kosuke and Keele, Luke and Tingley, Dustin},
	journal = {Psychol Methods},
	month = {Dec},
	number = {4},
	pages = {309--334},
	title = {A general approach to causal mediation analysis.},
	volume = {15},
	year = {2010}}

@article{Imai2010b,
author = {Kosuke Imai and Luke Keele and Teppei Yamamoto},
title = {{Identification, Inference and Sensitivity Analysis for Causal Mediation Effects}},
volume = {25},
journal = {Statistical Science},
number = {1},
publisher = {Institute of Mathematical Statistics},
pages = {51 -- 71},
year = {2010},
doi = {10.1214/10-STS321},
URL = {https://doi.org/10.1214/10-STS321}
}

@article{Robins1992,
	author = {Robins, J M and Greenland, S},
	journal = {Epidemiology},
	month = {Mar},
	number = {2},
	pages = {143--155},
	title = {Identifiability and exchangeability for direct and indirect effects.},
	volume = {3},
	year = {1992}}

@article{Baron1986,
  title={The moderator--mediator variable distinction in social psychological research: Conceptual, strategic, and statistical considerations.},
  author={Baron, Reuben M and Kenny, David A},
  journal={Journal of personality and social psychology},
  volume={51},
  number={6},
  pages={1173},
  year={1986},
  publisher={American Psychological Association}
}

@article{Wright1934,
  title={The method of path coefficients},
  author={Wright, Sewall},
  journal={The annals of mathematical statistics},
  volume={5},
  number={3},
  pages={161--215},
  year={1934},
  publisher={JSTOR}
}

@inproceedings{
Avin2005,
author = {Avin, Chen and Shpitser, Ilya and Pearl, Judea},
title = {Identifiability of path-specific effects},
year = {2005},
publisher = {Morgan Kaufmann Publishers Inc.},
address = {San Francisco, CA, USA},
booktitle = {Proceedings of the 19th International Joint Conference on Artificial Intelligence},
pages = {357–363},
numpages = {7},
location = {Edinburgh, Scotland},
series = {IJCAI'05}
}

@article{Daniel2015,
	author = {Daniel, R M and De Stavola, B L and Cousens, S N and Vansteelandt, S},
	journal = {Biometrics},
	month = {Mar},
	number = {1},
	pages = {1--14},
	title = {Causal mediation analysis with multiple mediators.},
	volume = {71},
	year = {2015}}

@article{Vanderweele2009,
  title={Conceptual issues concerning mediation, interventions and composition},
  author={VanderWeele, Tyler J and Vansteelandt, Stijn},
  journal={Statistics and its Interface},
  volume={2},
  number={4},
  pages={457--468},
  year={2009},
  publisher={International Press of Boston}
}

@article{Tchetgen2012,
author = {Eric J. Tchetgen Tchetgen and Ilya Shpitser},
title = {{Semiparametric theory for causal mediation analysis: Efficiency bounds, multiple robustness and sensitivity analysis}},
volume = {40},
journal = {The Annals of Statistics},
number = {3},
publisher = {Institute of Mathematical Statistics},
pages = {1816 -- 1845},
year = {2012},
doi = {10.1214/12-AOS990},
URL = {https://doi.org/10.1214/12-AOS990}
}

@article{VanderWeele2013b,
	author = {VanderWeele, Tyler J},
	journal = {Epidemiology},
	month = {Jan},
	number = {1},
	pages = {175--176},
	title = {Policy-relevant proportions for direct effects.},
	volume = {24},
	year = {2013}}

@article{Knafl2017,
	author = {Knafl, George J. and Knafl, Kathleen A. and Grey, Margaret and Dixon, Jane and Deatrick, Janet A. and Gallo, Agatha M.},
	journal = {BMC Medical Research Methodology},
	number = {1},
	pages = {45},
	title = {Incorporating nonlinearity into mediation analyses},
	volume = {17},
	year = {2017}}

@article{Robins2003,
  title={Semantics of causal DAG models and the identification of direct and indirect effects},
  author={Robins, James M},
  journal={Highly structured stochastic systems},
  pages={70--82},
  year={2003},
  publisher={Oxford University PressOxford}
}

@book{Vanderweele2015,
  title={Explanation in causal inference: methods for mediation and interaction},
  author={VanderWeele, Tyler},
  year={2015},
  publisher={Oxford University Press}
}

@article{Wright1921,
  author = {Wright, Sewall},
  journal = {Journal of agricultural research},
  number = 7,
  pages = {557--585},
  publisher = {Washington},
  title = {Correlation and causation},
  volume = 20,
  year = 1921
}

@article{Ding2016,
	author = {Ding, Peng and Vanderweele, Tyler J},
	journal = {Biometrika},
	month = {Jun},
	number = {2},
	pages = {483--490},
	title = {Sharp sensitivity bounds for mediation under unmeasured mediator-outcome confounding.},
	volume = {103},
	year = {2016}}

@article{TchetgenTchetgen2012,
author = {Eric J. Tchetgen Tchetgen and Ilya Shpitser},
title = {{Semiparametric theory for causal mediation analysis: Efficiency bounds, multiple robustness and sensitivity analysis}},
volume = {40},
journal = {The Annals of Statistics},
number = {3},
publisher = {Institute of Mathematical Statistics},
pages = {1816 -- 1845},
year = {2012},
doi = {10.1214/12-AOS990},
URL = {https://doi.org/10.1214/12-AOS990}
}

@book{Hernan2023,
	author = {Hern\'{a}n, M.A. and Robins, J.M.},
	publisher = {CRC Press},
	series = {Chapman \& Hall/CRC Monographs on Statistics \& Applied Probab},
	title = {Causal Inference: What If},
	year = {2023}}

@article{Xia2022,
	author = {Xia, Fan and Chan, Kwun Chuen Gary},
	journal = {Biometrika},
	month = {02},
	number = {4},
	pages = {1085-1100},
	title = {{Decomposition, identification and multiply robust estimation of natural mediation effects with multiple mediators}},
	volume = {109},
	year = {2022}}

@article{Miles2019,
	author = {Miles, C H and Shpitser, I and Kanki, P and Meloni, S and Tchetgen Tchetgen, E J},
	journal = {Biometrika},
	month = {11},
	number = {1},
	pages = {159-172},
	title = {{On semiparametric estimation of a path-specific effect in the presence of mediator-outcome confounding}},
	volume = {107},
	year = {2019}}

@book{Abraham2012,
  title={Manifolds, tensor analysis, and applications},
  author={Abraham, Ralph and Marsden, Jerrold E and Ratiu, Tudor},
  volume={75},
  year={2012},
  publisher={Springer Science \& Business Media}
}

@article{Zhou2023,
	author = {Zhou, Xiang and Yamamoto, Teppei},
	journal = {The Journal of Politics},
	number = {1},
	pages = {250-265},
	title = {Tracing Causal Paths from Experimental and Observational Data},
	volume = {85},
	year = {2023}}

@article{Shpitser2013,
	author = {Shpitser, Ilya},
	journal = {Cognitive Science},
	number = {6},
	pages = {1011-1035},
	title = {Counterfactual Graphical Models for Longitudinal Mediation Analysis With Unobserved Confounding},
	volume = {37},
	year = {2013}}

@article{Shpitser2008,
author = {Shpitser, Ilya and Pearl, Judea},
title = {Complete Identification Methods for the Causal Hierarchy},
year = {2008},
issue_date = {6/1/2008},
publisher = {JMLR.org},
volume = {9},
journal = {J. Mach. Learn. Res.},
month = {jun},
pages = {1941–1979},
numpages = {39}
}

@inproceedings{Malinsky2019,
  title={A potential outcomes calculus for identifying conditional path-specific effects},
  author={Malinsky, Daniel and Shpitser, Ilya and Richardson, Thomas},
  booktitle={The 22nd International Conference on Artificial Intelligence and Statistics},
  pages={3080--3088},
  year={2019},
  organization={PMLR}
}

@article{Efron1979,
	author = {B. Efron},
	journal = {The Annals of Statistics},
	number = {1},
	pages = {1 -- 26},
	title = {{Bootstrap Methods: Another Look at the Jackknife}},
	volume = {7},
	year = {1979}}

@book{Li2007,
  title={Nonparametric econometrics: theory and practice},
  author={Li, Qi and Racine, Jeffrey Scott},
  year={2007},
  publisher={Princeton University Press}
}

@article{Wang2016,
  title={Assessing natural direct and indirect effects for a continuous exposure and a dichotomous outcome},
  author={Wang, Wei and Zhang, Bo},
  journal={Journal of statistical theory and practice},
  volume={10},
  number={3},
  pages={574--587},
  year={2016},
  publisher={Springer}
}

@book{Silverman2018,
  title={Density estimation for statistics and data analysis},
  author={Silverman, Bernard W},
  year={2018},
  publisher={Routledge}
}

@article{Xu2021,
  title={Multiply robust causal mediation analysis with continuous treatments},
  author={Xu, Yizhen and Sani, Numair and Ghassami, AmirEmad and Shpitser, Ilya},
  journal={arXiv preprint arXiv:2105.09254},
  year={2021}
}

@article{Zhang2025,
  title={Doubly robust inference on causal derivative effects for continuous treatments},
  author={Zhang, Yikun and Chen, Yen-Chi},
  journal={arXiv preprint arXiv:2501.06969},
  year={2025}
}

@article{Huber2020,
  title={Direct and indirect effects of continuous treatments based on generalized propensity score weighting},
  author={Huber, Martin and Hsu, Yu-Chin and Lee, Ying-Ying and Lettry, Layal},
  journal={Journal of Applied Econometrics},
  volume={35},
  number={7},
  pages={814--840},
  year={2020},
  publisher={Wiley Online Library}
}

@article{Huang2024,
  title={Nonparametric estimation of mediation effects with a general treatment},
  author={Huang, Lukang and Huang, Wei and Linton, Oliver and Zhang, Zheng},
  journal={Econometric Reviews},
  volume={43},
  number={2-4},
  pages={215--237},
  year={2024},
  publisher={Taylor \& Francis}
}

@article{Dukes2023,
  title={Proximal mediation analysis},
  author={Dukes, Oliver and Shpitser, Ilya and Tchetgen Tchetgen, Eric J},
  journal={Biometrika},
  volume={110},
  number={4},
  pages={973--987},
  year={2023},
  publisher={Oxford University Press}
}

@article{Rudolph2021,
  title={Causal mediation with instrumental variables},
  author={Rudolph, Kara E and Williams, Nicholas and Diaz, Ivan},
  journal={arXiv preprint arXiv:2112.13898},
  year={2021}
}

@article{Huber2026,
  title={Difference-in-differences for mediation analysis using double machine learning},
  author={Huber, Martin and Oberh{\"a}nsli, Sarina Joy},
  journal={arXiv preprint arXiv:2602.23877},
  year={2026}
}

\end{document}